\documentclass[11pt]{article}

\usepackage[utf8]{inputenc}
\usepackage[T1]{fontenc}

\usepackage{microtype}

\usepackage[a4paper,margin=1in]{geometry}
\usepackage{amsmath}
\usepackage{amssymb}
\usepackage{mathrsfs}
\usepackage{textgreek}
\usepackage{graphicx}
\graphicspath{{./}{./figures/}{./figs/}{../}{../figures/}{../figs/}{../../}{../../figures/}{../../figs/}}
\usepackage{caption}
\usepackage{subcaption}
\usepackage{tikz-cd}
\usepackage{mathtools}
\usepackage{xspace}
\usepackage{tcolorbox}
\usepackage[normalem]{ulem}

\usepackage{xcolor}
\definecolor{darkgreen}{rgb}{0,0.4,0}
\definecolor{darkred}{rgb}{0.4,0,0}
\definecolor{darkblue}{rgb}{0,0,0.5}

\usepackage{hyperref}
\usepackage{doi}
\usepackage[capitalize, nameinlink]{cleveref}

\hypersetup{
    colorlinks=true,
    linkcolor=darkred,
    citecolor=darkgreen,
    filecolor=black,
    urlcolor=darkblue,}

\usepackage{authblk}
\usepackage{orcidlink}
\newcommand{\myorcid}[1]{{\small \orcidlink{#1}}}

\usepackage{enumitem}

\usepackage{algorithm}
\usepackage[noend]{algpseudocode} \makeatletter
\algnewcommand{\LineComment}[1]{\Statex \hskip\ALG@thistlm \(\triangleright\) #1}
\algnewcommand{\Statey}{\Statex \hskip\ALG@thistlm}
\makeatother

\algblockx[foreach]{ForEach}{EndForEach}
[1]{\textbf{for each} #1 \textbf{do}}
{\textbf{end}}

\usepackage{tabularray}

\usepackage{amsthm}
\usepackage{thmtools}
\usepackage{thm-restate}

\Crefname{remark}{Remark}{Remarks}
\Crefname{observation}{Observation}{Observations}

\theoremstyle{plain}
\newtheorem{theorem}{Theorem}[section]
\newtheorem{lemma}[theorem]{Lemma}
\newtheorem{proposition}[theorem]{Proposition}

\theoremstyle{definition}
\newtheorem{definition}[theorem]{Definition}

\theoremstyle{plain}

\theoremstyle{remark}
\newtheorem{remark}[theorem]{Remark}

\DeclareMathOperator{\poly}{poly}

\DeclareMathOperator{\dist}{dist}

\DeclareMathOperator{\ID}{ID}

\DeclarePairedDelimiter{\set}{\lbrace}{\rbrace}
\DeclarePairedDelimiter{\abs}{\lvert}{\rvert}
\DeclarePairedDelimiter{\card}{\lvert}{\rvert}

\DeclarePairedDelimiter{\ceil}{\lceil}{\rceil}

\newcommand{\calA}{\mathcal A}

\newcommand{\calL}{\mathcal L}

\newcommand{\bbN}{\mathbb{N}}

\newcommand{\bbR}{\mathbb{R}}
\newcommand{\bbZ}{\mathbb{Z}}

\newcommand{\naturals}{\bbN}

\newcommand{\reals}{\bbR}
\newcommand{\integers}{\bbZ}

\newcommand{\algoname}[1]{\textnormal{{\textsc{#1}}}}\newcommand{\CarefulWalk}{\algoname{CarefulWalk}}
\newcommand{\PathRulingSet}{\algoname{PathRulingSet}}
\newcommand{\EarlyStopPathRulingSet}{\algoname{ES-ColPathRulingSet}}
\newcommand{\ZWalk}{\algoname{ZWalk}}
\newcommand{\SearchingWalk}{\algoname{SearchingWalk}}

\newcommand{\lmin}{\ell_{\min}}
\newcommand{\lmax}{\ell_{\max}}

\newcommand{\rmout}{{\mathrm{out}}}
\newcommand{\Soutput}{S_\rmout}
\newcommand{\rmcand}{{\geq R}}
\newcommand{\Scand}{U^\rmcand}

\DeclareMathOperator{\care}{care}

\newcommand{\tstar}{{t^{*}}}

\newcommand{\Rstar}{{R^{*}}}

\newcommand{\Trdv}{{T_{\mathsf{rdv}}}}

\title{Optimal Deterministic Rendezvous in Labeled Lines}

\author[1]{Yann Bourreau \myorcid{0009-0001-1819-8348}}
\author[1]{Ananth Narayanan \myorcid{0009-0002-6137-4025}}
\author[2]{Alexandre Nolin \myorcid{0000-0002-3952-0586}}
\affil[1]{CISPA Helmholtz Center for Information Security, Saarbrücken, Germany}
\affil[2]{SAMOVAR, Télécom SudParis, Institut Polytechnique de Paris, Palaiseau, France}
\date{}

\begin{document}

\maketitle

\begin{abstract}
    In a \emph{rendezvous} task, a set of mobile agents initially dispersed in a network have to gather at an arbitrary common site.
    We consider the rendezvous problem on the infinite labeled line, with $2$ initially asleep agents, without communication, and a synchronous notion of time.
    Each node on the line is labeled with a unique positive integer.
    The initial distance between the two agents is denoted by $D$.
    Time is divided into rounds and measured from the moment an agent first wakes up.
    We denote by $\tau$ the delay between the two agents' wake up times.
If awake in a given round $T$, an agent at a node $v$ has three options: stay at the node $v$, take port $0$, or take port $1$. If it decides to stay, the agent will still be at node $v$ in round $T+1$. Otherwise, it will be at one of the two neighbors of $v$ on the infinite line, depending on the port it chose.
The agents achieve rendezvous in $T$ rounds if they are at the same node in round $T$.
We aim for a deterministic algorithm for this problem.

    The problem was recently considered by Miller and Pelc~\cite[Distributed Computing 2025]{MP_dc25}. With $\lmax$ the largest label of the two starting nodes, they showed that no algorithm can guarantee rendezvous in $o(D \log^* \lmax)$ rounds. 
    The lower bound follows from a connection with the LOCAL model of distributed computing, and holds even if the agents are guaranteed simultaneous wake-up ($\tau = 0$) and are told their initial distance $D$.
    Miller and Pelc also gave an algorithm of optimal matching complexity $O(D \log^* \lmax)$ when the agents know $D$, but only obtained the higher bound of $O(D^2 (\log^* \lmax)^3)$ when $D$ is unknown to the agents.

    In this paper, we improve this second complexity to a tight $O(D \log^* \lmax)$, closing the gap between the best known lower and upper bounds. In fact, our algorithm achieves rendezvous in $O(D \log^* \lmin)$ rounds, where $\lmin$ is the smallest label within distance $O(D)$ of the two starting positions.
    We obtain this result by having the agents compute sparse subsets of the nodes to gather at (formally, ruling sets over the line), as well as some general observations about the setting of rendezvous on labeled graphs.
\end{abstract}

\newpage
\setcounter{tocdepth}{2}
\tableofcontents
\newpage

\section{Introduction}

A classic task for a set of autonomous mobile agents is to meet at some location.
The agents might want to gather for a number of reasons: to exchange some resources, to initiate a long journey together, or to assemble in a specific shape.
With only two agents, the problem is known as \emph{rendezvous}~\cite{AG_book03}.

We consider a fully discrete version of this problem: time is divided into rounds, and in each round, each agent is located at a node of a graph $G$.  
Furthermore, we assume that each node $v$ of this graph has a unique label $\ell(v) \in \naturals^*$, and that the graph $G$ inhabited by the agents is an infinite line.
The agents wake up at possibly different rounds. Once awake, in each round, an agent can stay at its current location for one round or move to a node neighboring its current location.
Rendezvous is achieved if the agents are at the same node in some round. In particular, it does not suffice that the agents are in adjacent nodes, or cross the same edge in opposite directions: they lack a communication or sensing method to detect these situations.
The initial distance between the two agents is denoted $D$, and $\tau$ denotes the delay between the two wake-up times. 
The two agents execute the same algorithm, also known as the \emph{symmetric} setting.
We more formally describe our model in \cref{sec:preliminaries}.

\subsection{Our Results}

We give an optimal deterministic algorithm for the rendezvous problem on the infinite labeled line.

\begin{restatable}{theorem}{MainResultThm}
\label{thm:main-result}
    There is a deterministic algorithm for solving rendezvous on the infinite labeled line in $O(D\log^* \lmin)$ rounds from the wake-up time of the first agent, where $D$ is the initial distance between the two agents and $\lmin$ is the smallest label within distance $O(D)$ from the agents' starting locations.
\end{restatable}

\Cref{thm:main-result} improves upon a prior algorithm of complexity $O(D^2 (\log^* \lmax)^3)$ by Miller and Pelc~\cite{MP_dc25}, where $\lmax$ is the largest label of the two starting nodes of the agents.
To a lesser extent, we also improve upon their $O(D \log^* \lmax)$ algorithm that assumes that the agents know their initial distance $D$, by replacing their dependency in $\lmax$ by one in $\lmin$.

The optimality of our algorithm follows from a lower bound from the same paper by Miller and Pelc~\cite{MP_dc25}.
They proved that any algorithm running in $o(D \log^* \lmax)$ rounds would imply an $o(\log^*n)$ algorithm for $3$-coloring paths and cycles of $\poly(n)$ nodes in the LOCAL model.
A seminal result by Linial states that no such $o(\log^*n)$ algorithm exists~\cite{linial92}.
Note that our algorithm, while of complexity $o(D \log^* \lmax)$ if considering an infinite family of instances where $\lmin \in o(\lmax)$, is compatible with this lower bound, as families of instances where $\lmin \in \Theta(\lmax)$ can be constructed.
That is, our algorithm is only of complexity $o(D \log^* \lmax)$ on specific instances, not in general\footnote{Throughout this paper and as was the case in the paper by Miller and Pelc, we take a slightly non-standard definition of the $\log^*$ function guaranteeing that $\log^*(x)\geq 1$ for all $x$. This is to clarify that our $O(D \log^* \lmin)$ algorithm does not finish in $0$ rounds when $\lmin = 1$, since $\log^*(1) = 0$ with the most common definition of $\log^*$. See the Preliminaries (\cref{sec:preliminaries}).}.

Our algorithm for the infinite line immediately implies a similar result for finite paths and cycles.
\begin{restatable}{corollary}{FiniteGraphThm}
\label{thm:finite-graph-main}
    There is a deterministic algorithm for solving rendezvous on finite labeled paths and cycles in $O(\min(n,D\log^* \lmin))$ rounds from the wake-up time of the first agent, where $D$ is the initial distance between the two agents, $\lmin$ is the smallest label within distance $O(D)$ from the agents' starting locations, and $n$ is the number of nodes in the graph.
\end{restatable}

Like our algorithm for the infinite line, this improves upon previous results of complexity $O(\min(n,D^2 (\log^* \lmax)^3))$ and $O(\min(n,D \log^* \lmax))$ by Miller and Pelc~\cite{MP_dc25}.

\subsection{Technical Overview}
\label{sec:tech-overview}

\paragraph{Rendezvous and treasure hunt.}
To get some intuition on how rendezvous algorithms typically work, it is instructive to consider one of the simplest variants of the problem: the asymmetric setting with simultaneous startup.
Asymmetry means that each agent can execute its own strategy, or equivalently, one agent receives label $1$ while the other receives label $2$.
Simultaneous startup means that the two agents wake up in the same round.

These assumptions allow for the following strategy: one agent goes through successive phases exploring increasingly bigger intervals centered on its initial node, while the other agent waits at its initial node. Having the active agent explore intervals doubling in size between phases, the strategy achieves rendezvous in $O(D)$ rounds, which is optimal up to the hidden constant and lower order terms.
Rendezvous reduces to \emph{treasure hunt} in that scenario.

While this algorithm cannot be used as is in harder settings, our goal will be to somehow reduce to this simple strategy.
In the symmetric setting but with access to randomness, for example, the agents can flip coins to decide whether to search or stay put during some time intervals.
In all likelihood, the agents should regularly sample opposite roles.
Similarly, if the agents are given distinct labels (but importantly, not restricted to the set $\set{1,2}$), distinct behaviors can still be constructed from the given identifiers. The agents might for example alternate searching and inactive phases according to the binary representations of their labels.
More generally, we need the agents to \emph{break symmetry} between themselves, so that at some point during the execution, one of the two agents does a search while the other agent stays idle. Given that one agent's idle phase overlaps sufficiently with the other agent's searching phase, the two agents will achieve rendezvous.

\paragraph{How knowledge of $D$ helps.}
Having the agents know the initial distance $D$ between them drastically reduces the space of possibilities.
For instance, from the perspective of either agent, there are only two possibilities for the starting location of the other agent: the two nodes at distance exactly $D$ from its own initial node. 
If each agent spends its first $3D$ rounds to go to these two nodes, the first awake agent always finds the other agent in these $3D$ rounds, granted that the delay $\tau$ between the agents' wake-up times is at least $3D$.
Not meeting the other agent in these initial $3D$ rounds thus immediately implies a bound on the delay, $\tau < 3D$.

On the infinite labeled line, consider now the subset of nodes whose distances to the starting location of one of the agents are divisible by $D$.
The resulting set of nodes is the same whether we measure the distances from either of the two starting locations, as well as contains them.
Connecting each member of this set to the two other members closest to it on the original line, we obtain another infinite labeled line.
A possible strategy for the two agents to break symmetry is to agree on a coloring of this sub-line, and to act according to the colors of their starting nodes.
As the agents perform searching phases interspersed with waiting phases according to these colors, knowing $D$ allows the agents to limit their phases to intervals of $O(D)$ rounds.

\paragraph{Coloring an infinite line.}
Without additional assumptions, the agents cannot truly agree on a coloring of the whole sub-line. Indeed, the agents can only base the coloring of the sub-line on the labels of the nodes of the infinite line.
In order for the agents to be able to compute the color of a node, said color should only depend on a finite set of labels.
As the agents cannot explore more than $T$ nodes in $T$ rounds, for the sake of speed, the colors given to the agents' starting nodes should ideally be computable from labels in relative proximity to them.

At a high level, we want a mapping from some universe of possible local neighborhoods to colors, such that if applying this mapping to the infinite line, a valid coloring is obtained.
This idea of constructing a mapping between local neighborhoods and a node's local output almost corresponds to the LOCAL model of distributed computing.
A key difference, however, is that the standard LOCAL model assumes a finite graph.

Still, it is easy to adapt a seminal LOCAL algorithm for $3$-coloring paths and cycles to obtain a useful mapping.
To do so, the algorithm is transformed into an \emph{early-stopping} one, which has nodes terminate earlier depending on their given labels. While the classic LOCAL algorithm for $3$-coloring paths and cycles completes in $O(\log^* n)$ rounds, assuming that the nodes are given unique identifiers between $1$ and $\poly(n)$, the early-stopping version instead guarantees that a node $v$ of label $\ell(v)$ permanently commits to a color after $O(\log^* \ell(v))$ rounds.
From the point of view of mapping local views to colors, this means that there exists a mapping that infers the color of a node of label $\ell(v)$ from its distance-$O(\log^* \ell(v))$-neighborhood.
Used on the infinite sub-line obtained by taking one out of every $D$ nodes, this allows the two agents to assign distinct colors to their starting nodes after $O(D \log^* \lmax)$ rounds of exploration, where $\lmax$ is the largest of the starting locations' labels. This roughly summarizes the approach of Miller and Pelc~\cite{MP_dc25}.

\paragraph{Limitations of the regular sub-line approach.}
The approach sketched in previous paragraphs is very dependent on the assumption that the agents have knowledge of $D$, as it has the agents essentially ignore all but a specific $1/D$ fraction of the line.
If the agents run the same algorithm with a wrong ``guess'' for $D$, the two agents are effectively not synchronizing in any way, as the sub-lines they each construct are only the same if the agents' guess for $D$ is a divisor of the actual distance $D$.
Even when their guess divides $D$, the agents' starting nodes are only adjacent in the sub-line if the guess is exactly $D$.

As a result, to tackle the setting without knowledge of $D$ by using the previous algorithm with different guesses for $D$ basically requires to test \emph{every possible value} for $D$.
This is true even if the agents are given an upper bound on $D$.
This difficulty is the source of the $O(D^2)$ dependency in the earlier algorithm by Miller and Pelc~\cite{MP_dc25}.

\paragraph{Limitations of a purely coloring-based approach.}
For simplicity, let us assume for now that the agents are given an upper bound on $D$. In the paper, we will see that handling a fully unknown $D$ can be handled relatively easily by essentially considering geometrically increasing bounds on $D$.
As a first try, let us aim for the agents to agree on a common coloring of the line that ensures that their starting locations receive distinct colors.

With only an upper bound on $D$, the agents now need to break symmetry with $2D$ nodes instead of just $2$.
Guaranteeing that any two nodes within distance $D$ from one another are colored differently cannot be done with $3$ colors.
Pushing forward with this approach, one could color the line with $O(D)$ colors.
But now, the increased number of colors necessitates to design the same number of distinct behaviors of alternating searching and waiting phases guaranteeing rendezvous.
This higher number of behaviors is likely to lead to a $D \log D$ dependency in the distance, similar to how in anonymous networks with labeled agents, larger labels lead to higher complexity in this fashion\footnote{For instance, with $\ell$ the smallest of the two agents' labels, $D$ their initial distance, and simultaneous startup, the complexity of rendezvous is $\Theta(D \log \ell)$ in the anonymous ring \cite{DFKP_algorithmica06}.}.
This would be far from the $\Omega(D\log^* \ell)$ lower bound.

\paragraph{Our solution: ruling set nodes as local landmarks.}
To reduce the number of behaviors we have to design, and thus, improve on the complexity that would follow from an approach based on computing a distance-$D$ coloring, we have the agents agree on a $1/\Theta(D)$ fraction of the nodes that act as common milestones, or landmarks.
One way to understand our algorithm is that we have the agents \emph{build their common sparse sub-line} instead of receiving it from the get-go as when they know $D$.
Formally, the agents compute a \emph{ruling set} over the line: a set of nodes such that every two of its members are at least at some minimum distance $\Theta(D)$ from one another, and every node in the graph is within some maximum distance $\Theta(D)$ from the nearest member.
Computation of this ruling set by the agents corresponds, again, to a computation in the early-stopping LOCAL model.

To achieve rendezvous, each agent effectively substitutes its starting position with the nearest node on the sub-line.
By having the agents agree on a common coloring of this sub-line, as before, the agents can now find each other by performing some waiting and searching phases based on the colors of their location in the sub-line.
Even if our new sub-line is less structured than the one considered in the setting with knowledge of $D$, its properties are sufficient to be able to give distinct colors to members of the sub-line at distance $O(D)$ in the original line, while only using $O(1)$ colors.

Computing the ruling set we need in the early-stopping LOCAL model follows relatively straightforwardly from standard results in distributed computing.
We perform this computation by repeated applications of an algorithm for Maximal Independent Set (MIS).
At a high level, we start from a set containing all nodes of the lines and gradually remove nodes from it. Each call to the MIS subroutine roughly doubles the minimum distance between members of the set, while maintaining that the maximum distance of any node on the graph to the nearest member of the set increases only moderately.
A small subtlety is that we have each node mostly coordinate with other nodes whose labels are of the same order of magnitude.

\paragraph{A simplifying remark: avoiding simultaneous edge crossings.}
Contributing to the overall simplicity of our algorithm is a subroutine for crossing edges that prevents the agents from exchanging positions on the line without achieving rendezvous.
The idea behind it is very simple: on the labeled line, every edge can be given a canonical orientation, for example, from lower to higher ID.
By having the agents perform a different set of moves to cross an edge depending on the direction in which they take it, we can make sure that the agents must be at the same node at some point during the crossing.
The first few moves are the same for both orientations of the edge, which enables an agent to discover the orientation of the edge while crossing it -- e.g., on the labeled line, through learning the identifier of the other endpoint.
Replacing every one-round edge crossing by this procedure only incurs a multiplicative factor of $4$ in the complexity of the algorithm.
We believe that this observation might be useful to the design of future rendezvous algorithms on graphs whose edges can easily be given an orientation.

\subsection{Related Work}

Most relevant to the results of this paper is the work by Miller and Pelc~\cite{MP_dc25}, that considered the rendezvous problem with the exact same set of assumptions as we do in this paper, as well as some easier settings.
They showed that when agents know the distance $D$ between themselves, and the maximum labels of the two starting nodes of the agents is $\ell$, deterministic rendezvous can be done in $O(D \log^* \ell)$ rounds. They gave a matching lower bound of $\Omega(D \log^* \ell)$, shown on a line whose IDs all have the same $\log^* \ell$. The lower bound holds even with simultaneous startup. The upper bound does not assume this.
When $D$ is unknown, they give an upper bound of $O(D^2 (\log^* \ell)^3)$.
They also gave an algorithm to achieve rendezvous in $O(D)$ rounds when the agents know their positions on the infinite line, improving upon an earlier result by Collins et al.~\cite{CollinsCGKM11} which additionally assumed simultaneous startup of the agents.

On an anonymous line, where the agents themselves are given two distinct labels, the problem is known to be harder. Dessmark et al.~\cite{DFKP_algorithmica06} showed that the problem has complexity $\Theta(D \log \ell)$, assuming simultaneous startup, where $\ell$ is the smallest of the two agents' labels. With arbitrary startup, they consider a harder setting where the agents only appear in the graph once they are awake. 
With that harder assumption, the algorithms tend to have a dependency on $n$, the size if the graph (which is then not assumed to be infinite in this work). 
Rendezvous in infinite lines was first studied in~\cite{MARCO2006315}, in which the authors considered the harder asynchronous setting, but with the agents capable of detecting each other in edges.

\newcommand{\asFrom}[1]{\scalebox{0.5}{(#1)}}

\newcommand{\assumpFormat}[1]{{\sffamily #1}\xspace}
\newcommand{\asLin}{\assumpFormat{Lin}}
\newcommand{\asFin}{\ensuremath{n{\ll}\infty}\xspace}
\newcommand{\asNL}{\assumpFormat{NL}}
\newcommand{\asAL}{\assumpFormat{AL}}
\newcommand{\asSWU}{\assumpFormat{SWU}}
\newcommand{\asDbA}{\assumpFormat{DbA}}
\newcommand{\asDtO}{\assumpFormat{DtO}}
\newcommand{\asOL}{\assumpFormat{OL}}
\newcommand{\asPtO}{\assumpFormat{PtO}}
\newcommand{\asSoO}{\assumpFormat{SoO}}
\newcommand{\asEwA}{\assumpFormat{EwA}}
\newcommand{\asLoG}{\assumpFormat{LoG}}
\newcommand{\asNumber}{8}
\newcommand{\totalCols}{10}

\begin{table}[bt]
    \centering
\begin{tblr}[
        caption = {Overview of previous results},
label = {table:comparison}
    ]{
        colspec = { *\asNumber{X[1,c]} *2{c}},
        row{1} = {font=\large\bfseries\sffamily}, hspan=minimal, rowhead = 1, cell{1}{1} = {c=\asNumber}{c}, cell{11-11}{1} = {c=\totalCols}{j,cmd=\small},
        hline{1,3,11}={2pt},
        hline{2,4-11}={1pt},
        vline{1,Z} = {1-10}{2pt},
        vline{2-Y} = {1-10}{1pt},
        row{10} = {bg=black!10!white}
      }
      Assumptions & & & & & & & & \SetCell[r=2]{c} Bound & \SetCell[r=2]{c} Citation
      \\
      \asLin    & $\mathclap{\asFin}$     & \asNL      & \asAL      & \asEwA     & \asSWU     & \asDbA     & \asLoG     &
      \\
      \checkmark & \checkmark & \checkmark & \checkmark & \checkmark & \checkmark & \checkmark & \checkmark &
      $\Omega(D)$ & folklore
      \\
      \checkmark & \checkmark &            & \checkmark &            & \checkmark &             &             &
      $\Theta(D \log \ell)$ & \cite{DFKP_algorithmica06}
      \\
      \checkmark &            & \asFrom{\asLoG} & \asFrom{\asLoG}     & \checkmark & \checkmark &            & \checkmark &
      $O(D)$ & \cite{CollinsCGKM11}
      \\
      \checkmark &            & \asFrom{\asLoG} & \asFrom{\asLoG} & \checkmark &            &            & \checkmark &
      $O(D)$ & \cite{MP_dc25}
      \\
      \checkmark &            & \checkmark & \asFrom{\asNL} & \checkmark & \checkmark & \checkmark &             &
      $\Omega(D\log^* \ell)$ & \cite{MP_dc25}
      \\
      \checkmark &            & \checkmark & \asFrom{\asNL} & \checkmark &           & \checkmark &             &
      $\Theta(D\log^* \ell)$ & \cite{MP_dc25}
      \\
      \checkmark &  & \checkmark & \asFrom{\asNL} & \checkmark &  &  &             &
      $O(D^2(\log^* \ell)^3)$ & \cite{MP_dc25}
      \\
      \checkmark &  & \checkmark & \asFrom{\asNL} & \checkmark &  &  &             &
      $\Theta(D \log^* \ell)$ & \textbf{this paper}
      \\
      Acronyms for assumptions: 
      \asLin (Graph is Line-like): the graph has maximum degree $2$. \asFin (Finite Graph): when unchecked, the graph can be infinite.
      \asNL (Node Labels): each node is equipped with a unique label. Implied by \asLoG, and implies \asAL.
      \asAL (Agent Labels): the agents are each given a unique label. Implied by \asLoG and \asNL.
      \asEwA (Exists while Asleep): rendezvous is also achieved if an awake agent finds an asleep agent. 
      \asSWU (Simultaneous Wake-Up): the agents wake up in the same round ($\tau = 0$).
      \asDbA (Distance between Agents): the agents know the initial distance $D$ between their starting locations.
      \asLoG (Location on Graph): the graph has some global naming scheme, which is known to the agents, and in which they know where they are. Implies \asNL and \asAL.
\end{tblr}
    \caption{Comparison of our result with some of the most relevant works in the literature.}
    \label{tab:enter-label}
\end{table}

Before they were formulated as questions of computational complexity and algorithm design, coordination problems of agents trying to meet appeared in the economics literature~\cite{book_Schelling}.
Algorithms for rendezvous in graphs were first studied by Alpern~\cite{S0363012993249195}.
Some of the graph structures in which deterministic rendezvous has been considered include the infinite tree \cite{det_rv_trees}, the infinite grid \cite{det_rv_grid}, the plane \cite{det_rv_plane} and the ring \cite{KranakisSSK03}.
Deterministic rendezvous was also studied in general graphs~\cite{Ta_ShmaZ14}, and in a setting where agents are given additional information about their location in the graph~\cite{CollinsCGKM11}.
Rendezvous on the line has also been studied with random resources~\cite{S036301299427816X,S0363012996314130}.

The $\Omega(D \log^* \ell)$ lower bound obtained in \cite{MP_dc25} comes from a reduction of the rendezvous problem in labeled lines to the problem of coloring nodes on a line in the LOCAL model with some number of colors $C\geq 3$. This problem has been shown to have a $\Omega(\log^*n - \log^* C)$ lower bound \cite{linial92,LS_podc14_ba}.
This lower bound is tight, due to an earlier algorithm by Cole and Vishkin matching its complexity~\cite{ColeVishkin_iandc86}.
Many works consider a version of the LOCAL model in which the nodes know an upper bound $n$ on the number of nodes, and have IDs between $1$ and $\poly(n)$. 
A generic way to remove this type of global knowledge was developed in a work by Korman, Sereni, and Viennot~\cite{KormanSV_dc13}.
Several ideas are transferable from the setting where $n$ is unknown to that of infinite graphs.

Readers interested in a broader view of the area of deterministic rendezvous might be interested in a recent survey by Pelc~\cite{Pelc_survey_lncs19}.
The book by Alpern and Gal~\cite{AG_book03} offers an introduction to the area of rendezvous that includes many topics outside the scope of this paper (asymmetric agents, continuous versions of the problem).
For an introduction to the LOCAL model, see the book by Peleg~\cite{book_Peleg}, as well as the online book by Hirvonen and Suomela~\cite{HS_da2020book}.

\subsection{Organization of the Paper}

We start with some preliminaries in \cref{sec:preliminaries}, where we introduce the notations and the model we are going to be considering throughout the paper.

In \cref{sec:subroutines}, we present the techniques and algorithms that the agents are going to be using to break symmetry between them and achieve rendezvous. \cref{sec:carefulwalk} shows how to devise a way for the agents to meet during the execution of the algorithm when they both move through the same edge at the same time but from different starting nodes.
In \cref{sec:pathrulingset}, we build an efficient algorithm to compute ruling sets on paths in LOCAL. Then, using this algorithm, we obtain in \cref{sec:earlystop-pathrulingset} an efficient algorithm for computing ruling sets on paths in LOCAL in which nodes with smaller IDs terminate earlier. 

\Cref{sec:our-algorithm} describes an optimal algorithm for the agent to reach rendezvous in our setting and analyses its complexity. \Cref{sec:algorithm-description} presents the algorithm computed by the agents. To prove that the agents achieve rendezvous, we focus on a specific phase of the algorithm of one of the agents. We show that after reaching this phase the agents will achieve rendezvous before switching to the next phase.
We split the analysis according to whether the agents are in the same phase or not. We argue about the two agents achieving rendezvous when in different phases in \cref{sec:distinct-phase-rdv}, and when in the same phase in \cref{sec:same-phase-rdv}.
We additionally show that the phase of our algorithm for which our arguments hold is guaranteed to occur within our stated upper bound, and finish the proof of our optimal upper bound in \cref{sec:algorithm-analysis}.

We give in \cref{sec:other-results} the implications of our results for line-like graphs other than the infinite line: finite lines, finite cycles, and the infinite half-line (containing a single node of degree $1$).

\section{Preliminaries}
\label{sec:preliminaries}

\subsection{Notation}
Throughout the paper, $\log$ refers to the base-$2$ logarithm, i.e., $\log(2^x) = x$ for any $x \in \reals$.
For each $x \in \reals$, the $\log^*$ function is defined as $\log^*(x)=1$ if $x \leq 1$, $\log^*(x)=1+\log^*(\log(x))$.
This is off-by-one from the maybe more common definition that sets $\log^*(x)=0$ if $x \leq 1$.
We do so to have that $\Theta(D\log^* \ell)$ is of order $\Theta(D)$ rather than $0$ when $\ell = 1$.

\paragraph{Graph notation.}
We write $G=(V,E)$ for a graph over vertices $V$ and edges $E$.
Each node $v \in V$ is equipped with a unique label or identifier $\ell(v) \in \calL$.
We identify the set of labels with the positive integers, $\calL = \naturals^*$.
We partition the nodes according to their labels and the $\log^*$ function, defining 
$V_i = \set{v \in V: \log^*(\ell(v)) = i}$, so that $V = \bigsqcup_{i \geq 1} V_i$. We also define $V_{\leq i} = \bigsqcup_{j = 1}^i V_j$ and $V_{<i} = V_{\leq i} \setminus V_i$.
The degree of a node $v \in V$ is denoted by $\deg(v)$, the distance $\dist(v,v')$ between two nodes corresponds to the (edge-)length of a shortest path between them, and $N(v)$ is the set of neighbors of $v$. For every integer $k$, $N^k(v) = \set{u \in V : \dist(u,v) \in [1,k]}$ is the $k$-hop neighborhood of $v$.

For $U \subseteq V$ a subset of the nodes of a graph $G = (V,E)$, the graph $G[U] = (U,E')$ is the graph induced by $U$, whose set of nodes is $U$ and whose set of edges is $E' = \set{vv' \in E: v \in U \wedge v' \in U}$.
For a positive integer $k\geq 1$, the power graph of order $k$ is the graph $G^k=(V,E'')$ with the same set of nodes as $G$ and whose set of edges is $E'' = \set{vv': \dist_G(v,v') \in [1,k]}$.
Note that $G^k[U] = (G^k)[U]$ rather than $(G[U])^k$, i.e., $G^k[U]$ is an induced subgraph of a power graph of $G$, rather than a power graph of an induced subgraph of $G$.

\paragraph{The Labeled Infinite Line.}
An infinite labeled line is a $2$-regular connected graph $G=(V,E)$ over a countably infinite set of nodes.

For the analysis' sake, it is sometimes convenient to consider a canonical orientation and naming for the line.
When we do so, we refer to nodes on the infinite line as relative integers, s.t.\ for any $i \in \integers$, node $v_i$ is adjacent to nodes $v_{i-1}$ and $v_{i+1}$ on the line.

\subsection{Computational Models}

\paragraph{Our Mobile Agent Model.}
The agents are placed and move in a network modeled by a graph $G=(V,E)$. $G$ can be finite or countably infinite. Each node $v \in V$ of the graph has $\deg(v)$ ports, and each edge is the graph is matched to exactly one port of each of its two endpoints.
The ports of a node $v$ are labeled $0$ to $\deg(v)-1$.
At a node $v$, an agent can read the node's label $\ell(v)$ if the graph is labeled, see the node's degree $\deg(v)$, and distinguish the $\deg(v)$ ports leading out of the node.

The agents execute the same deterministic algorithm, as in, when put in the same configuration, agents $\alpha$ and $\beta$ perform the same set of moves.
Time is divided into discrete rounds.
The agents are initially asleep, and upon wake-up (or startup), they have no notion of how much time has elapsed before. Formally, the agents have an internal clock that is initialized at $0$ upon waking-up, and from there on keeps track of how many rounds passed since wake-up. We denote by $\tau$ the delay between the wake-up times of the two agents. The agents can be thought of as waking up in rounds $0$ and $\tau$.
The initial distance between the two agents is denoted by $D$.

In each round, an awake agent at a node $v$ chooses one of $\deg(v)+1$ possible moves. One of these moves is staying at $v$ for one round, the other moves being to take one of the ports leading out of $v$.
When an agent takes a port, in the next round, it arrives at the other endpoint of the edge behind that port. When entering a node, an agent knows the port it enters it from.
We put no limit on the memory of the agents, and they can base each of their decisions on everything they have experienced in previous rounds: the labels of all the nodes they were at in previous rounds, and from which port they were all entered from.

The agents can detect each other if they are at the same node in the same round, in which case they achieve rendezvous. 
This is also the case if one of the agents is asleep, i.e., agents exist in the network before waking up.
They have no other means of communication.
By default, the agents do not detect each other if they cross the same edge in opposite directions in the same round. We consider agents able to detect simultaneous edge-crossing in some parts of the paper.

The agents have no shared notion of orientation of the line, no idea of their position on the line with respect to some shared point of origin, and have no knowledge of both $D$ and $\tau$.

\paragraph{The LOCAL Model on Infinite Graphs.}
The standard LOCAL model is a model of distributed computation in which a network is represented as a graph $G=(V,E)$ whose nodes possess computational power.
Every node has a unique identifier, typically assumed to be an integer between $1$ and $\poly(n)$, where $n$ is the number of nodes in the graph.
The time is measured in rounds, in which a node can receive the messages sent by its neighbors in the previous round, do some internal computation, and send a message to each of its neighbors for them to receive in the next round.
Computation starts at the same time for all nodes and is done synchronously.
An algorithm's complexity on instances of size $n$ is the maximum number of rounds before all nodes in the graph have terminated executing the algorithm and permanently committed to an output.
The complexity of a problem on instances of size $n$ is the minimum complexity over all possible algorithms solving the problem over instances of size $n$.
Messages sent in the LOCAL model can be arbitrarily large. As a result, without loss of generality, each node learns its entire $T$-hop neighborhood in $T$ rounds of LOCAL.
This makes a $T$-round algorithm for a problem in the LOCAL model equivalent to a map from possible $T$-hop neighborhoods to possible node outputs.

In infinite graphs, it is nonsensical to consider the complexity of most problems as the time needed for the last node to terminate.
The complexity also cannot be measured with respect to the number of nodes or the value of the largest identifier.
To tame the infinite, 
we instead aim to bound the time it takes a node to finish with respect to some local parameters, for example, its identifier.
To have nodes terminate at different times, the problems we consider must be inherently quite local.
We refer to this model, where nodes commit to an output in a time that is a function of a local parameter, and thus algorithms make sense on both finite and infinite graphs,
as the EarlyStop-LOCAL model (ES-LOCAL for short).
This model was also used in the work by Miller and Pelc~\cite{MP_dc25} and appeared previously in other sources~\cite[Exercise 1.7]{HS_da2020book}.

\subsection{Useful Subroutines}

Consider a graph $G = (V,E)$ and a subset $S \subseteq V$ of its nodes. The set $S$ is \emph{independent} iff for all $u,v \in S$, $uv \not \in E$. $S$ is \emph{dominating} iff for all $u \in V \setminus S$, $\exists v \in S$, $uv \in E$.
If $S$ is both independent and dominating, it is a \emph{maximal independent set} (MIS).

Let each node $v\in V$ be given a list of colors $\psi(v)$.
A list-coloring of the graph for these lists is a map $\chi:V \to \bigcup_{v\in V} \psi(v)$ s.t.\ for each node $v$, $\chi(v) \in \psi(v)$, and for any two adjacent nodes $u$ and $v$, $\chi(u) \neq \chi(v)$.
When each node $v$ is given a list of size $\card{\psi(v)} \geq \deg(v)+1$, the coloring is known as a \emph{degree+1-list coloring}.
Having lists with this minimum size guarantees the existence of a solution $\chi$. In fact, it also guarantees that a partial solution can always be extended to a full one.
When the lists are all equal to $[k]$ for some integer $k$, the problem is known as $k$-coloring. In the extensively studied $\Delta+1$-coloring problem, each node has $[\Delta+1]$ as its list of colors, where $\Delta$ is an upper bound on the maximum degree of the graph.

Recall that in the LOCAL model, $n$ corresponds to the number of nodes in the graph, and $\poly(n)$ bounds the nodes' identifiers.

\begin{lemma}[\cite{linial92}]
\label{lem:linial}
    Both computing a Maximal Independent set and computing a degree+1-list coloring can be done in $O(\log^* n)$ rounds of LOCAL on $n$-node graphs of constant maximum degree $\Delta \in O(1)$.
\end{lemma}

\section{Our Main Methods}
\label{sec:subroutines}

\subsection{Ensuring Rendezvous from Crossing: CarefulWalk}
\label{sec:carefulwalk}

Our model assumes that the agents have no detection mechanism other than that they know if they are at the same node in a given round, in which case they achieve rendezvous. In particular, the agents cannot detect if they cross the same edge in opposite directions in the same round, and this does not count as achieving rendezvous. In this section, we show that having a labeled line makes this distinction irrelevant.

We give a simple procedure (\cref{alg:carefulwalk}) to cross edges which ensures that the agents cannot exchange positions on the line without achieving rendezvous. To do so, \CarefulWalk\ takes $4$ rounds to cross each edge. Its basic idea is to have the agents behave differently depending on whether they cross an edge in the direction of increasing or decreasing IDs. The same can be achieved given a common orientation of the edges instead of node labels.

\begin{algorithm}[ht]
\caption{\CarefulWalk, for crossing an edge $e = uv$}
\label{alg:carefulwalk}
\begin{algorithmic}[1]
    \Statex Let $u$ be the initial node, $v$ the target node, $e=uv$ the edge to cross.
    \State Wait at initial node for one round.\Comment{Agent at $u$ at $t=0,1$.}
    \State Cross the edge: $u \to v$. \Comment{Agent at $u$ at $t=1$, at $v$ at $t=2$.}
    \If{$\ID(v) > \ID(u)$}
\State Wait two rounds. \Comment{Agent at $v$ at $t=2,3,4$.}
    \Else \Comment{$\ID(v) < \ID(u)$}
    \State Cross the edge: $v \to u$. \Comment{Agent at $v$ at $t=2$, at $u$ at $t=3$.}
    \State Cross the edge: $u \to v$. \Comment{Agent at $u$ at $t=3$, at $v$ at $t=4$.}
    \EndIf
\end{algorithmic}
\end{algorithm}

\begin{figure}
        \centering
        \small
        \includegraphics[page=1,width=0.15\textwidth]{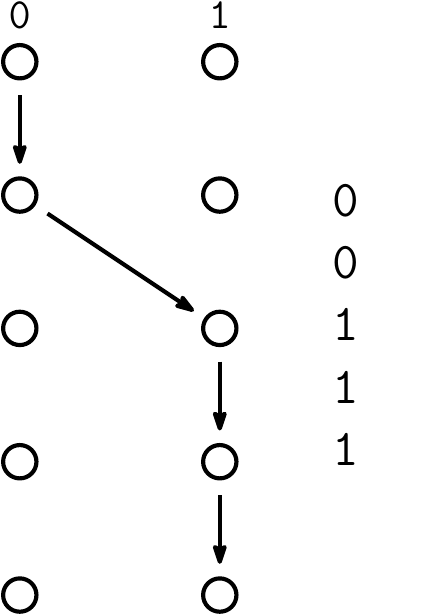}
        \hspace{0.1\textwidth}
        \includegraphics[page=2,width=0.15\textwidth]{careful_walk_v1}
        \hspace{0.1\textwidth}
        \includegraphics[page=3,width=0.15\textwidth]{careful_walk_v1}
        \caption{Illustration of the moves performed by \CarefulWalk\ to cross an edge $uv$. The five circles below each $0$ (resp.\ $1$) represent the node of lower ID (resp.\ higher ID) at $5$ successive time steps. Arrows represent the moves (or absence thereof) of an agent in a \CarefulWalk\ crossing the edge. The left and middle figures depict crossing the edge in opposite directions. The right figure depicts two agents performing a \CarefulWalk\ over the same edge simultaneously in opposite directions, achieving rendezvous at time step $4$, after $3$ moves. The $5$-bit strings in the figures describe the positions occupied the agents over the $5$ time steps.}
        \label{fig:careful_walk}
\end{figure}

Consider an algorithm $\calA$ for the moves of an agent. Let $\care(\calA)$ be the algorithm obtained from $\calA$ by replacing each edge-crossing in $\calA$ by a call to \CarefulWalk, and stretching each stay at a node by a factor of $4$. This new algorithm basically does the same moves as $\calA$ at a slower pace.
We show in \cref{thm:no-edge-crossing} that this transformed algorithm is able to achieve rendezvous as $\calA$ would in the stronger setting where the agents can detect simultaneous edge-crossing.

\begin{theorem}
\label{thm:no-edge-crossing}
    Let $G$ be a labeled graph, with agents placed at two locations $v_\alpha$ and $v_\beta$ of the graph, $T$ an integer, and let $\calA$ be a deterministic algorithm for the agents' moves.
    Suppose that for any delay $\tau\geq 0$ between the wake-up times of the two agents, with simultaneous edge-crossing counting as rendezvous, the two agents achieve rendezvous in at most $T$ rounds when executing $\calA$ on the graph $G$ with $v_\alpha$ and $v_\beta$ as the starting nodes for the two agents.

    Let $\care(\calA)$ be
    the algorithm obtained from $\calA$
by replacing every edge-crossing by a \CarefulWalk, and every action of staying at a node by $4$ such actions.
    Then in the setting without simultaneous edge-crossing, on the same graph $G$ and with the same starting nodes $v_\alpha$ and $v_\beta$, $\care(\calA)$ has the agents achieve rendezvous in at most $4T$ rounds, for any delay $\tau\geq 0$ between the agents' wake-up times.
\end{theorem}
\begin{proof}
    Let agent $\alpha$ be the first agent to wake up (in round $0$), and $\beta$ the other agent (waking up in round $\tau \geq 0$).
    Consider the execution of $\care(\calA)$. In that execution, for each integer $i \geq 0$, let $x(i)$ be the position of agent $\alpha$ in round $i$, and $y(i)$ that of agent $\beta$. In particular, $y(i) = y(0)$ for all $i \leq \tau$.
    
    Let $\tau' = \ceil{\tau / 4}$ and $r = 4 \tau' - \tau$.
    Let $x'(0) = x(0)$, $y'(0) = y(0)$, and for all $i \geq 1$, let $x'(i) = x(4i)$ and $y'(i) = y(4i - r)$.
    Note that the rounds $\set{4i : i\in \naturals}$ (resp.\ $\set{4i-r: i \geq \tau'}$) are exactly the rounds in which agent $\alpha$ (resp.\ $\beta$) finishes and starts a $4$-round action in $\care(\calA)$ -- staying put for $4$ rounds, or doing a \CarefulWalk.
    As a result, $x'(i)$ and $y'(i)$ are the positions that agents $\alpha$ and $\beta$ would have in round $i$ of an execution of the original algorithm $\calA$ with a delay of $\tau'$. 

    If agents start at the same node, rendezvous is trivially achieved immediately. Otherwise, since $\calA$ is guaranteed to achieve rendezvous in $T$ or less rounds (with simultaneous edge-crossings counting as rendezvous), there must exist some $t \in [1,T]$ such that either $x'(t) = y'(t)$, or $x'(t) = y'(t-1)$ and $y'(t) = x'(t-1)$. Let $\tstar\geq 1$ be the smallest such value.
    
    Suppose that $x'(\tstar) = y'(\tstar)$. If $r = 0$ (i.e., $\tau = 0 \mod 4$), then $\alpha$ and $\beta$ are at the same node $x(4\tstar) = y(4\tstar)$ in round $4\tstar$. 
    Otherwise ($r \in \set{1,2,3}$), consider the positions of the agents in rounds $4\tstar$ and $4\tstar-2$.
    Regardless of whether $\alpha$ was staying idle for $4$ rounds or performing a \CarefulWalk\ to node $x(4\tstar)$ in the $4$ rounds prior to round $4\tstar$, it was at node $x(4\tstar)$ in rounds $4\tstar$ and $4\tstar-2$ (see \cref{alg:carefulwalk} and \cref{fig:careful_walk}, the description and an illustration of \CarefulWalk).
    Similarly, regardless of the action that agent $\beta$ initiates in round $4\tstar-r$, it is at node $y(4\tstar-r) = x(4\tstar)$ in rounds $(4\tstar-r)$ and $(4\tstar-r+1)$. For all possible values of $r$, the agents meet in either round $4\tstar$ or $4\tstar-2$.
    
    Suppose that $x'(\tstar) = y'(\tstar-1)$ and $x'(\tstar-1) = y'(\tstar)$.
    A similar case analysis as the one we just did shows that the agents necessarily meet in a round in the interval $[4\tstar-4,4\tstar-1]$ (see \cref{fig:rendezvous-careful-walk}).

    \begin{figure}[htb]
        \centering
        \small
        \newcommand{\CarefulFigureWidth}{0.1\textwidth}
        \newcommand{\CarefulFigureGapWidth}{0.01\textwidth}
\subcaptionbox{\label{subfig:-3-gap}}{\includegraphics[page=4,width=\CarefulFigureWidth]{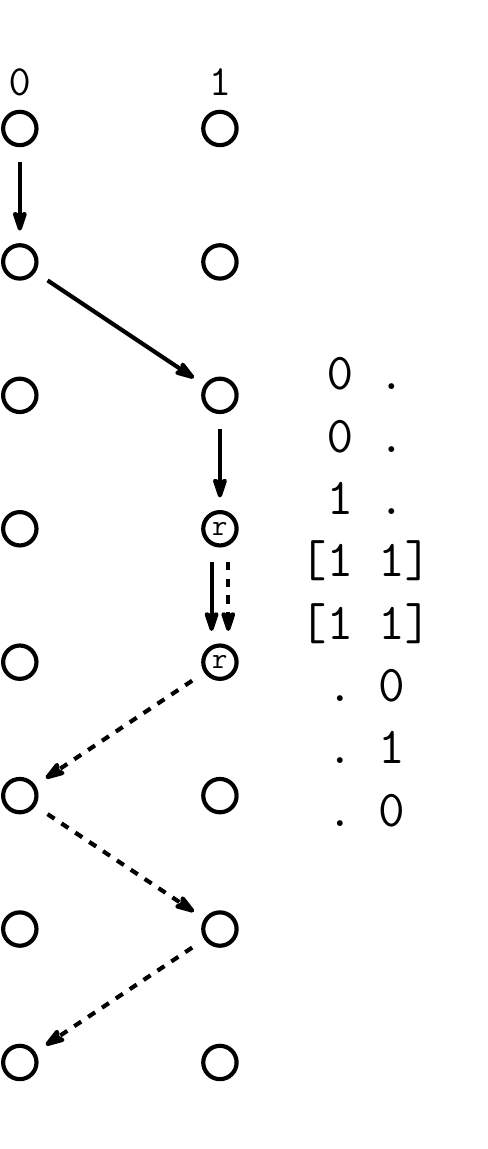}}\hspace{\CarefulFigureGapWidth}\subcaptionbox{\label{subfig:-2-gap}}{\includegraphics[page=5,width=\CarefulFigureWidth]{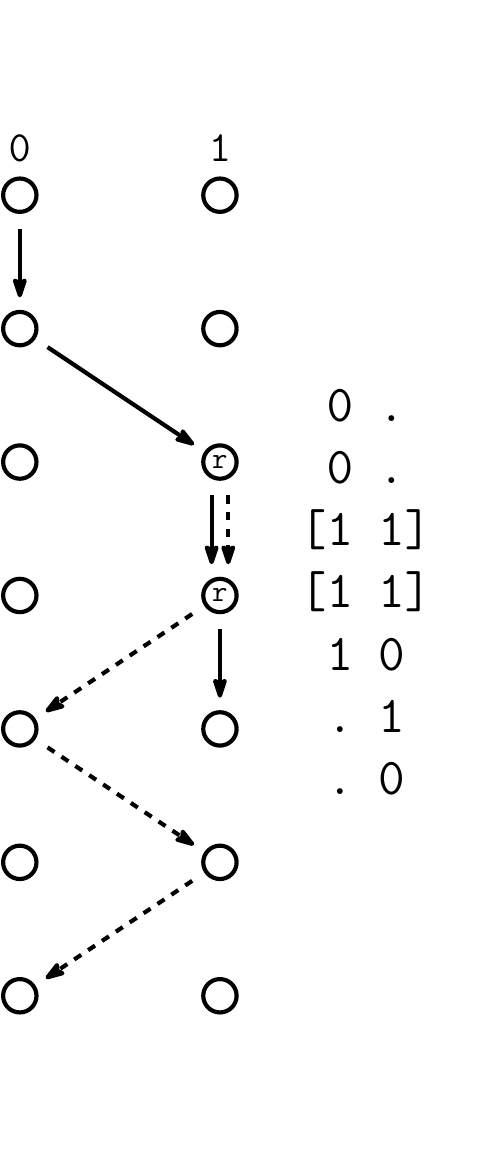}}\hspace{\CarefulFigureGapWidth}\subcaptionbox{\label{subfig:-1-gap}}{\includegraphics[page=6,width=\CarefulFigureWidth]{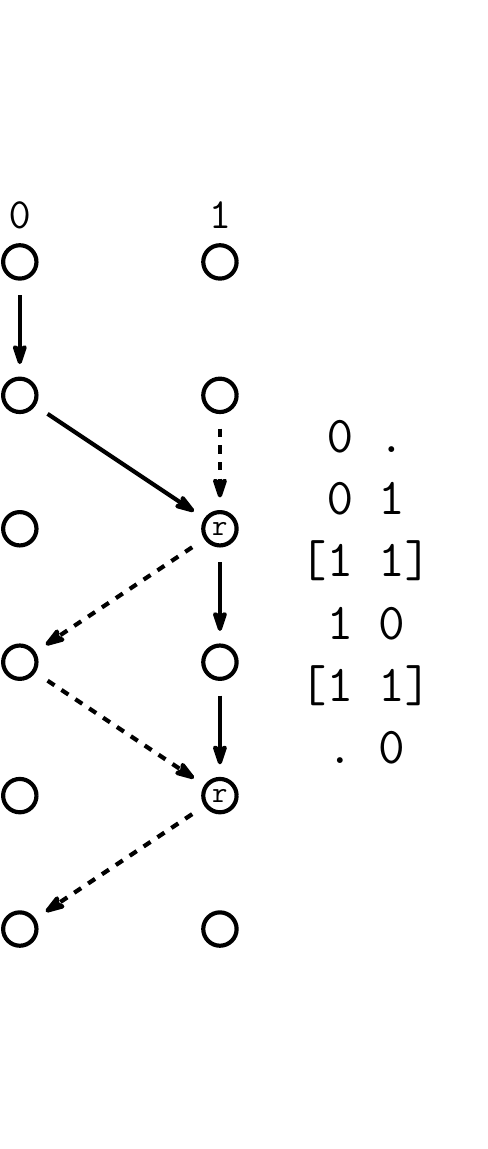}}\hspace{\CarefulFigureGapWidth}\subcaptionbox{\label{subfig:0-gap}}{\includegraphics[page=7,width=\CarefulFigureWidth]{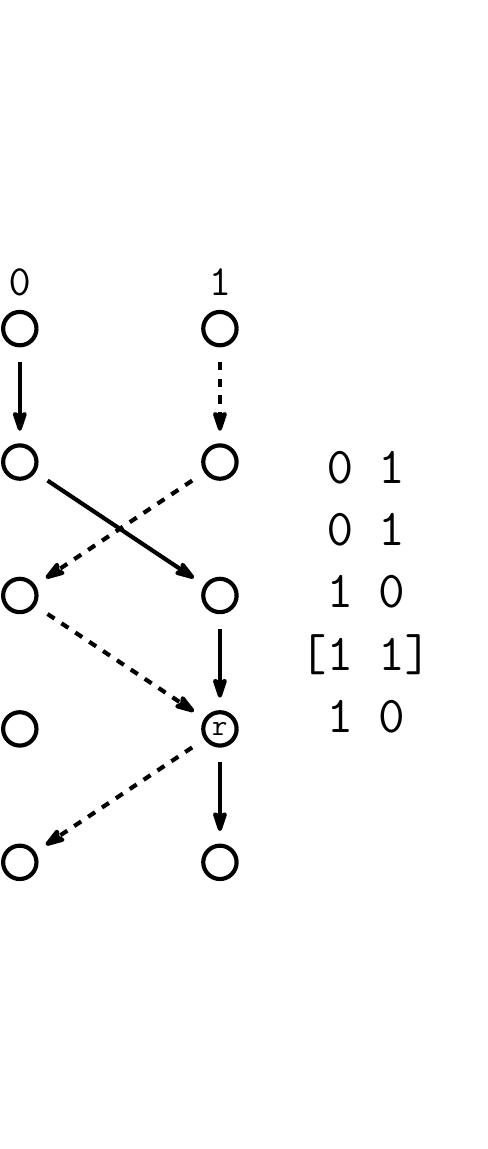}}\hspace{\CarefulFigureGapWidth}\subcaptionbox{\label{subfig:1-gap}}{\includegraphics[page=8,width=\CarefulFigureWidth]{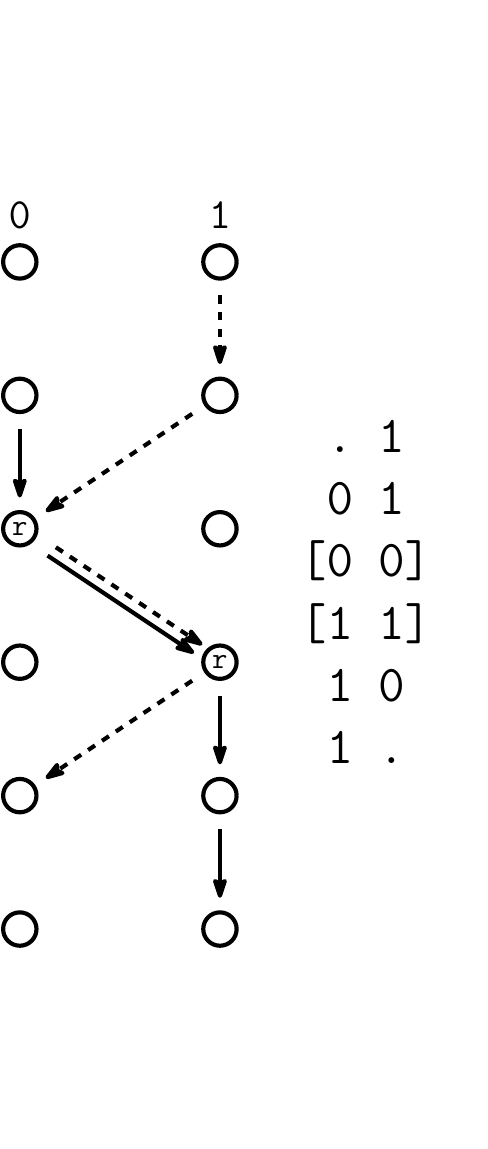}}\hspace{\CarefulFigureGapWidth}\subcaptionbox{\label{subfig:2-gap}}{\includegraphics[page=9,width=\CarefulFigureWidth]{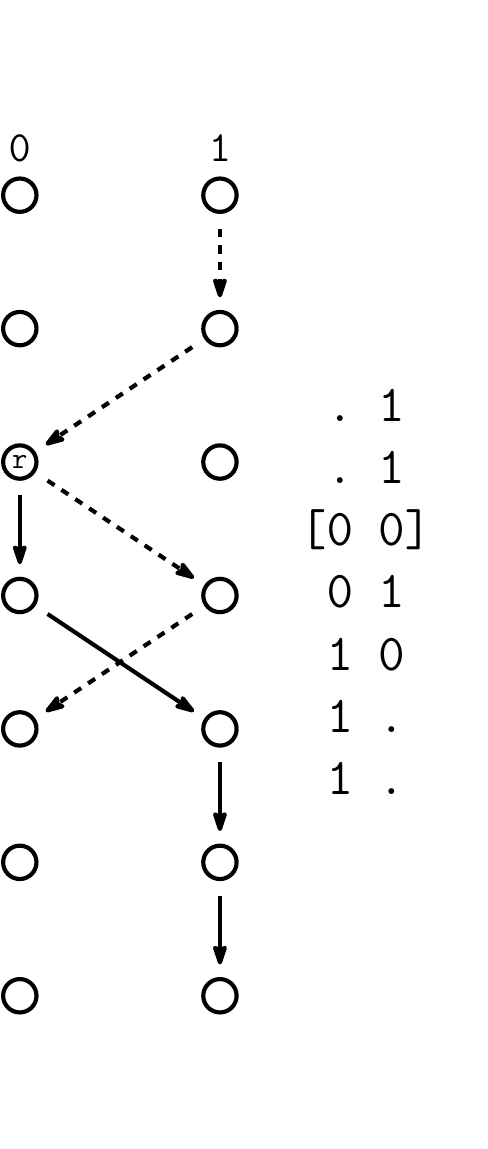}}\hspace{\CarefulFigureGapWidth}\subcaptionbox{\label{subfig:3-gap}}{\includegraphics[page=10,width=\CarefulFigureWidth]{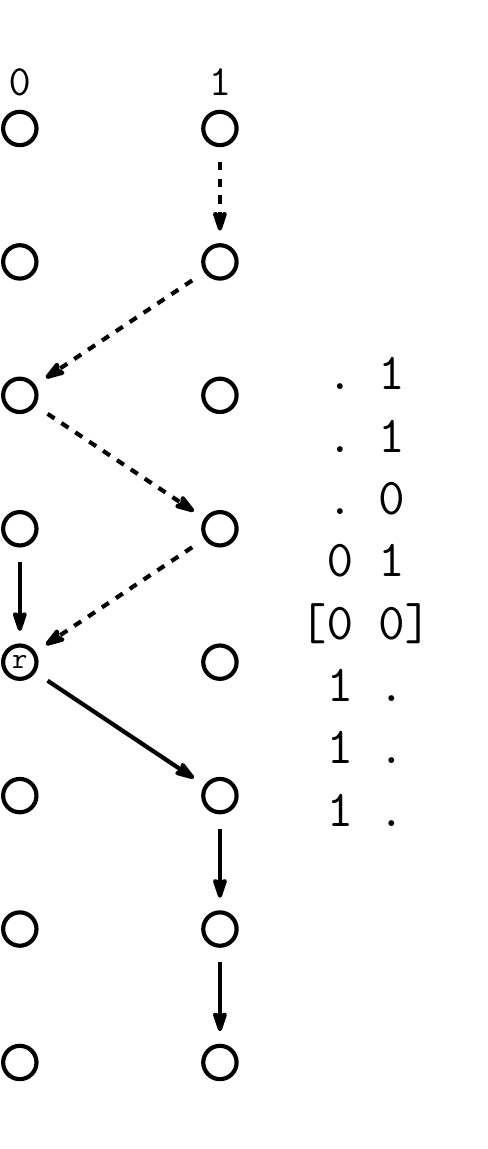}}\caption{The agents always meet when crossing the same edge in opposite directions using \CarefulWalk.}
        \label{fig:rendezvous-careful-walk}
    \end{figure}

    In all cases, agents executing $\care(\calA)$ achieve rendezvous in round $4T$ at the latest, proving the claim.
\end{proof}

\Cref{thm:no-edge-crossing} allows us to design algorithms for the simpler setting where agents detect simultaneous edge-crossing, and automatically translate rendezvous results about these algorithms to the setting where agents only achieve rendezvous if at the same node in a given round, which is the original setting of interest.
This dramatically simplifies the work of the algorithm designer.
Importantly, on the line, having a detection mechanism for simultaneous edge-crossing prevents the agents to swap positions without achieving rendezvous.

\begin{proposition}
    \label{prop:no-position-exchange}
    Consider a canonical naming of the nodes on the infinite line as relative integers, s.t.\ for any $i \in \integers$, node $v_i$ is adjacent to nodes $v_{i-1}$ and $v_{i+1}$ on the line.
Let $a,b,c,d \in \integers$ be integers and $T,T'$ be two rounds with $T' > T$. For each $t\in[T,T']$, let $x(t) \in \integers$ (resp.\ $y(t) \in \integers$) be the location of agent $\alpha$ (resp.\ $\beta$) in round $t$. 
    Consider the following two scenarios:
    \begin{enumerate}
        \item (Agents swap positions on the line) $a\leq b<c\leq d$, $x(T) = a < c = x(T')$ and $y(T) = d > b = y(T')$. 
        \item (An agent scans an interval that the other agent stays in) $a \leq b \leq c \leq d$, $\forall z \in [a,d]$, $\exists t \in [T,T']$ s.t.\ $x(t) = z$, and $\forall t \in [T,T']$, $y(t) \in [b,c]$.
    \end{enumerate}
    If simultaneous edge-crossing counts as rendezvous, then in both scenarios, the agents necessarily meet in a round between $T$ and $T'$.
\end{proposition}

\subsection{Computing Ruling Sets on Paths}
\label{sec:pathrulingset}

In this section, we give a LOCAL algorithm for computing what is essentially a ruling set over a path, but focusing on a subset of the nodes $U \subseteq V$.
\begin{definition}[Set-Limited Ruling Sets]
    \label{def:subset-ruling-set}
    Let $\alpha$ and $\beta$ be two nonnegative integers, $G=(V,E)$ be a graph, and $U \subseteq V$ a subset of its nodes.
    A set of nodes $S$ is a $U$-limited $(\alpha,\beta)$-ruling set of $G$ if and only if:
    \begin{enumerate}
    \item $S \subseteq U$,
    \item For any two nodes $u,v \in S$, $\dist(u,v) \geq \alpha$ (packing property),
    \item For any node $u \in U$, $\exists v \in S$, $\dist(u,v) \leq \beta$ (covering property).
    \end{enumerate}
\end{definition}
A standard ruling set of a graph $G = (V,E)$ is a $V$-limited ruling set of $G$. As simple examples of standard ruling sets, the whole set of nodes $V$ is an $(1,0)$-ruling set, a maximal independent set of a graph $G$ is a $(2,1)$-ruling set of that graph, and a maximal independent set of its order-$k$ power graph $G^k$ is a $(k+1,k)$-ruling set of $G$.

Our $U$-limited ruling sets are exactly like standard ruling sets, except that they only need to cover nodes in $U$ and can only contain nodes from $U$. Importantly, they are not equivalent to computing a standard ruling set of $G[U]$: in our definition of $U$-limited ruling sets, we measure the distances in the original graph $G$, not $G[U]$.
Taking a finite set $U$ on an infinite graph $G$ means that its complexity can be measured as the longest time taken by a node to terminate, as in the standard LOCAL model, even as the communication takes place on an infinite graph.
This algorithm will be the basis of our subsequent algorithm that performs a similar task, but additionally colors the ruling set and is early stopping.

\begin{algorithm}[ht]
\caption{\PathRulingSet.}
\label{alg:pathrulingset}
\begin{algorithmic}[1]
    \Statex \textbf{Input:} A finite set of nodes $U \subseteq V$ with IDs in $[n]$, an integer distance $R \geq 1$.
    \Statex Let $d = \ceil{\log R}$, $S_{0} \gets U$. \Comment{The ruling set is initially the entire set $U$}
    \For{$i = 1$ to $d$} \Comment{Each loop iteration computes $S_i \subseteq S_{i-1}$}
    \State \textbf{if} $i \leq d-1$ \textbf{then} Let $R_{i} \gets 2^{i} - 1$ \textbf{else} Let $R_{i} \gets R - 1$.
    \State Consider the graph $H_{i} \gets G^{R_{i}}[S_{i-1}]$. \State Compute an MIS over $H_{i}$, let $S_{i}$ be its nodes.
    \EndFor
\State Initialize $\Soutput \gets S_d$.
    \For{$6$ iterations}
    \State For each $v \in U$, compute $b_v = \min_{u\in \Soutput} \dist(v,u)$. Let $\Scand = \set{v \in U: b_v \geq R}$.
\State Add to $\Soutput$ all $v\in \Scand$ s.t.\ $\forall u \in N^{R-1}(v) \cap \Scand, (b_u < b_v \vee (b_u = b_v \wedge \ell(v) > \ell(u)))$.
    \EndFor
        
\Statex \textbf{Return:} $\Soutput$.
\end{algorithmic}
\end{algorithm}
\newcommand{\PRSdl}{\kappa_{\ref{alg:pathrulingset}}}
\newcommand{\PRSdd}{\kappa'_{\ref{alg:pathrulingset}}}
\newcommand{\PRSc}{\kappa''_{\ref{alg:pathrulingset}}}

\begin{lemma}
\label{lem:pathrulingset}
    Executed by a set of nodes $U$ with identifiers between $1$ and $n$ on an infinite labeled line, \cref{alg:pathrulingset} computes a $U$-limited $(R,R-1)$ ruling set of $G$ in $O(R \log^* n)$ rounds of deterministic LOCAL.
\end{lemma}

\begin{proof}
    $d = \ceil{\log R}$ implies that $2^{d} \geq R > 2^{d-1}$.
    We show that the following properties always hold for each integer $i \in [0,d-1]$: the nodes in the set $S_i$ are at distance at least $2^{i}$ from each other (packing property), and every node from $U$ is at distance at most $2^{i+1}-i-2$ from a node in $S_i$ (covering property).
    We then analyze the distances for the last two sets $S_{d}$ and $\Soutput$.
    Finally, we use those properties to guarantee the $O(R \log^*n)$ runtime.

    \paragraph{Packing property.}
    For $i=0$, the property trivially holds: $2^{i} = 1$ and two distinct nodes are necessarily at distance at least $1$ from each other. For $i \in [1,d-1]$, the property follows directly from the fact that in iteration $i$ of the loop, we compute a maximal independent set over $H_{i} = G^{R_{i}}[S_{i-1}]$. Since the nodes in $S_{i}$ form an independent set in $G^{R_{i}}$, any two nodes in $S_{i}$ are not adjacent in $G^{R_{i}}$, thus are at distance at least $R_{i}+1 = 2^{i}$ from each other in $G$, giving the claim. Additionally, nodes in $S_d$ are at least $R_d+1 = R$-apart by the same argument.
    Finally, the nodes that are added to $S_d$ to obtain $\Soutput$ are selected to be at distance $\geq R$ from other nodes already in $\Soutput$ (condition $b_v \geq R$ in computing $\Scand$), and so that two nodes at distance $< R$ are never simultaneously added to $\Soutput$ (condition that $b_v$ is maximal among $\Scand$ nodes within a $(R-1)$-radius, using IDs as tie-breakers)

    \paragraph{Covering property.}
    Since $S_0 = U$, for $i = 0$, we have that every node in $U$ is at distance at most $0 = 2^{i+1}-i-2$ from a node in $S_i$.

    Consider an integer $i \in [1,d-1]$ and a node $v \in U$.
    Let $u \in S_{i-1}$ be the node that is nearest to $v$ in the set $S_{i-1}$. If this node is still in $S_{i}$, then $\dist(v,S_{i}) = \dist(v,S_{i-1})$. If not, then $u$ has a neighbor $u' \in G^{R_i}[S_{i-1}]$ that is still in the next independent set $S_{i}$, meaning that $\dist(u,u')\leq R_i$.
    Thus, for any node $v \in U$, $\dist(v,S_{i}) \leq \dist(v,S_{i-1}) + R_i$.
    
    Therefore for every integer $i \in [1,d-1]$, the maximum distance of a node in $U$ to a node in $S_{i}$ is at most $\sum_{j=1}^{i} R_j = \sum_{j=1}^{i} (2^j -1)= 2^{i+1} - i - 2 \leq 2^{d} - d - 1 \leq 2R - 2$.
    The MIS computation over $S_d$ adds at most $R-1$ to this distance, giving a maximum distance of $3R-3$.
    
    Finally, consider the computation of $\Soutput$, initialized as $S_d$. 
    Consider an iteration of the second loop, and take the current values of $\Soutput$ and $(b_v)_{v \in U}$.
    Consider a node $v \in U$ s.t.\ $b_v \geq R$. A node in its distance-$(3R-4)$ neighborhood must be added to $\Soutput$. Indeed, we have $b_u \leq 3R-3$ for all nodes $u \in U$ from the established properties of $S_d$.
    This implies that within distance $2R-3$ from $v$, there is a node $w$ s.t.\ $b_w \geq R$ and $b_w \geq b_u$ for all $u \in N^{R-1}(w)$.
    Either this node is added to $\Soutput$, or a node $u \in N^{R-1}(w)$ of equal value $b_u = b_w$ is added to $\Soutput$. 

    Suppose that after $k$ iterations of the loop, for some node $v \in U$, it holds that $b_v \geq R$. This means that we had $b_v \geq R$ for the $k$ iterations of the loop, and therefore, $k$ nodes were added to $\Soutput$ in the distance-$(3R-4)$ neighborhood of $v$. A subsequent iteration number $k+1$ should add another node to $\Soutput$ within this neighborhood. But nodes in $\Soutput$ are at distance at least $R$ from one another, so the distance-$(3R-4)$ neighborhood of $v$ can contain at most $6$ nodes in $\Soutput$. Therefore, after $6$ iterations of the loop, all nodes $v \in U$ satisfy $b_v \leq R-1$, i.e., are at distance $< R$ from $\Soutput$.

    \paragraph{Complexity.}
    $G$ being a path/line, the fact that two nodes in $S_{i-1}$ are at distance at least $2^{i-1} \geq (R_i+1)/2$ from one another in $G$ bounds the maximum degree of $H_i = G^{R_i}[S_{i-1}]$ to $2$.
The nodes of these graphs are also all from $U$, which bounds their identifiers to be at most $n$.
    These two bounds imply that all the computations we do on each of these graphs can be done in $O(\log^* n)$ rounds of LOCAL on the said graph.
Simulating a round of LOCAL on a graph $H \subseteq G^k$ can be done in $k$ rounds of communication over $G$.
    As we consider $\Theta(\log R)$ graphs whose costs of simulation increase geometrically, we get a total complexity of
$O(\sum_{i=1}^{d} R_i \log^* n)= O(R \log^* n)$ for the $d$ MIS computations.
    The greedy extension of $S_d$ to $\Soutput$ only adds $O(R)$ rounds to this complexity, as $O(R)$ rounds suffice for each node to compute its distance $\leq 3R-3$ to the nearest node in $S_d$ and compare it to the values of the nodes within distance $R-1$.
\end{proof}

\subsection{Computing Colored Ruling Sets on Paths, and Stopping Early}
\label{sec:earlystop-pathrulingset}

The previous section considered the standard LOCAL model of distributed computing, applied to a finite set of nodes (at most $n$, due to the bound on the nodes' identifiers), even if on an infinite line. In this section, we use this algorithm to compute the same kind of object on the infinite line, in the Early Stop paradigm.

Our algorithm gradually computes a ruling set $S$ of the infinite line, in the sense that we provide an algorithm in which every node on the infinite line eventually decides whether it is part of the ruling set or not, and in any given round, the set of nodes that committed to their output and joined the set $S$ is a ruling set over the nodes that committed to their output.
At a high level, \cref{alg:earlystoppathrulingset} partitions the nodes on the infinite line according to their IDs, and computes a ruling set over each set of nodes of this partition using \cref{alg:pathrulingset} (\PathRulingSet) from the previous section.
The computation done by the set of nodes $V_i = \set{v\in V: \log^*\ell(v) = i}$ is well defined as it is performed by a finite set of nodes, with bounded IDs. After $\Theta(R \cdot i)$ rounds, the nodes in $V_i$ have all finished their execution of \PathRulingSet, as did previously nodes in the sets $V_j$, $j<i$.
Once $V_i$ has computed a $V_i$-limited ruling set $S_i$, we use it to extend the current ruling set $S$, that only contains nodes from $V_{<i}$
at this point.
This is done greedily. We first have each node $v \in S_i$ join $S$ if and only if $S$ does not already contain a node at distance $<R$ from $v$ on the line. A second greedy extension slightly improves the covering distance. Finally, we assign a color to the new members of the ruling set.

\begin{algorithm}[ht]
\caption{\EarlyStopPathRulingSet.}
\label{alg:earlystoppathrulingset}
\begin{algorithmic}[1]
    \Statex \textbf{Input:} An infinite labeled line $G$, an integer distance $R \geq 1$.
    \State $S \gets \emptyset$ current ruling set, $\chi : S \to [17]$ current coloring.
    \ForEach{$i \in \naturals^*$ in parallel}
        \State Compute a $V_i$-limited $(R,R-1)$-ruling set with \PathRulingSet.
\Statey Let $S_i$ be the set and $\chi_i$ be the coloring of $S_i$ computed on the previous line.
        \State Wait for some $\Theta(R\cdot i)$ rounds, ensuring that nodes in $V_{<i}$ all chose their output.
\Statey Let $S_{<i}$ be the value of $S$ at this stage. 
        \Comment{$S_{<i}$ is a $V_{<i}$-limited $(R,R-1)$-ruling set}
\State Let $S'_i \gets S_i \setminus N^{R-1}(S)$.
        \State Update $S \gets S \cup S'_i$.
        \For{$4$ iterations}
            \State For each $v \in V_i$, compute $b_v = \min_{u\in S} \dist(v,u)$. Let $\Scand_i = \set{v \in V_i: b_v \geq R}$.
            \State Add to $S$ all $v\in \Scand_i$ s.t.\ $\forall u \in N^{R-1}(v) \cap \Scand_i, (b_u < b_v \vee (b_u = b_v \wedge \ell(v) > \ell(u)))$.
        \EndFor
        \State Each $u \in S \setminus S_{<i}$ computes $\psi(u) = [17] \setminus \set{\chi(w): w \in S_{<i} \cap N^{9R-1}(u)})$.
        \State Compute a list-coloring of $G^{9R-1}[S \setminus S_{<i}]$ according to the lists $(\psi(u))_{u \in S \setminus S_{<i}}$. \State Extend $\chi$ from $S_{<i}$ to $S$ with this coloring.
        \State Each node $v \in V_i$ learns $N^{R-1}(v) \cap S$ and $\chi(u)$ for each $u \in N^{R-1}(v) \cap S$.
        \State Nodes in $V_i$ commit to their output.
    \EndForEach
\end{algorithmic}
\end{algorithm}

\newcommand{\ESPRSdl}{\kappa_{\ref{alg:earlystoppathrulingset}}}
\newcommand{\ESPRSdd}{\kappa'_{\ref{alg:earlystoppathrulingset}}}
\newcommand{\ESPRSc}{\kappa''_{\ref{alg:earlystoppathrulingset}}}

\begin{lemma}
\label{lem:earlystoppathrulingset}
    For any round $T$, let $U_T$ be the set of nodes that committed to their outputs by round $T$, $S_T \subseteq U_T$ the set of nodes that joined $S$ by round $T$, and $\chi_T: S_T \to [17]$ the colors assigned to nodes in $S_T$.
    \Cref{alg:earlystoppathrulingset} achieves the following:
    \begin{enumerate}
        \item A node $v$ of label $\ell(v)$ terminates in $O(R \log^* \ell(v))$ rounds,
        \item At termination, a node $v$ outputs whether it is part of the ruling set $S$ or not, and if part of $S$, it also outputs a color between $1$ and $17$. If outside $S$, $v$ knows the members of $S$ at distance $\leq R - 1$ from itself together with their colors.
        \item $S_T$ is a $U_T$-limited $(R,R-1)$-ruling set, and $\chi_T$ is a proper coloring of $G^{9R-1}[S_T]$. 
\end{enumerate}
\end{lemma}
\begin{proof}
    We first show that \cref{alg:earlystoppathrulingset} (\EarlyStopPathRulingSet) constructs the claimed object,
    before bounding its complexity.

    \paragraph{Properties of the Ruling Set.}
    Recall that $S_i$, the output of \PathRulingSet\ executed by nodes in $V_i$, is a $V_i$-limited $(R,R-1)$-ruling set.
    Consider the set $S_{<i}$, the content of $S$ before adding nodes from $V_i$ to it.
    Assume $S_{<i}$ to be a $V_{<i}$-limited $(R,R-1)$-ruling set. We show that the nodes from $V_i$ added to it make $S$ a $V_{\leq i}$-limited $(R,R-1)$-ruling set. By induction, starting from the trivial base case of $S_{<1}$, we get the claimed properties about our ruling set.
    
    Every node $v \in V_i$ is at distance at most $R-1$ from a node $u \in S_i \subseteq V_i$. Consider the state of $S$ after adding $S'_i = S_i \setminus N^{R-1}(S)$ to it. If $u$ was added to $S$, then $v$ remains at distance at most $R-1$ from a node in $S$. If $u$ was not added to $S$, there exists a node $w \in S_{<i}$ that joined $S$ earlier, such that $\dist(u,w) < R$. This implies that $v$ is at distance at most $\dist(v,u) + \dist(u,w) \leq 2R-2$ from a node in $S$. At this state, $S = S_{<i} \cup S'_i$ is a $V_{\leq i}$-limited $(R,2R-2)$-ruling set. 
    The $4$ iterations of the inner loop that add to $S$ some of the nodes in $V_i$ furthest away from the current set $S$ make $S$ a $V_{\leq i}$-limited $(R,R-1)$-ruling set, from similar arguments as used in the proof of \cref{lem:pathrulingset}.

    Finally, the minimum distance of $R$ between nodes in $S$ ensures a maximum degree of $16$ for the graph $G^{9R-1}[S]$.
    Each list $\psi(u)$ for $u \in S\setminus S_{<i}$ therefore contains at least $\deg(u) + 1$ colors, where $\deg(u)$ is the degree of $u$ in $G^{9R-1}[S]$. As the lists are chosen so as to avoid creating conflicts with previously colored nodes in $S$, the extended coloring $\xi$ is a valid coloring of $G^{9R-1}[S]$.

    \paragraph{Complexity.}
    Recall that 
\PathRulingSet\ (\cref{alg:pathrulingset}) 
    with parameter distance $R\geq 1$
    has a complexity of $O(R \log^* n)$ rounds
when run by nodes of maximum ID $n$.
Nodes in $V_i$ thus spend $O(R\cdot i)$ rounds computing $S_i$ at the beginning of the (parallel) loop in \cref{alg:earlystoppathrulingset}.

    For every $i \in \naturals^*$, let $T_i$ be the round in which nodes in $V_i$ terminate in \cref{alg:earlystoppathrulingset}.
    Nodes in $V_i$ first compute $S_i$, before greedily adding nodes of $S_i$ to $S$.
    Before adding nodes of $S_i$ to $S$, all nodes in $V_j, j<i$ must have terminated, which requires waiting until round $T_{i-1}$ ($0$ for the first set of nodes).
    After this amount of time, inserting the nodes of $S_i$ into $S$ only takes $R-1$ rounds, for each node in $S_i$ to learn which nodes in its $(R-1)$-neighborhood have joined $S$.
    The $4$ loop iterations in which additional nodes from $V_i$ are added to $S$ are also completed in $O(R)$.
    Finally, solving the list-coloring problem is done in $O(R \cdot i)$ rounds (\cref{lem:linial}), and each node learning the output of nodes within distance $R$ takes $O(R)$ rounds.
    
In total, we have that
$T_i \leq \max(O(R \cdot i), T_{i-1}) + O(R \cdot i)$.
    By induction, we get that $T_i \in O(R \cdot i)$.
\end{proof}

\section{Our Optimal Rendezvous Algorithm}
\label{sec:our-algorithm}

\subsection{The Algorithm}
\label{sec:algorithm-description}

In this section, we describe our deterministic rendezvous algorithm. We actually describe an algorithm working in the setting where agents achieve rendezvous when crossing the same edge in opposite directions. By 
\cref{thm:no-edge-crossing}, this implies an algorithm with similar asymptotic complexity for the setting where agents only detect each other if at the same node in a given round.

Recall that both agents behave according to the same algorithm, which is just a sequence of moves. In most of this section, we take the perspective of one of the two agents, which we will denote as agent $\alpha$. We call agent $\beta$ the other agent. $v_\alpha$ and $v_\beta$ are as before the starting nodes of agents $\alpha$ and $\beta$.

The two basic building blocks of our algorithm are \ZWalk\ and \EarlyStopPathRulingSet, respectively \cref{alg:doubling-walk} and \cref{alg:earlystoppathrulingset}.
\ZWalk\ is simply a procedure in which the agent executing it scans the $2L+1$ nodes of the distance-$L$ neighborhood of the node it starts the \ZWalk\ from, in $4L$ rounds.
Note that while the agents do not have a common orientation, they have a sense of direction in that they know the port from which they enter a node. This allows the agents to keep walking in the same direction for successive steps, and to change direction, as desired.
Most moves taken during our algorithm are done in executions of \ZWalk.
\EarlyStopPathRulingSet\ is the early-stopping algorithm for computing colored ruling sets presented in the previous section.

\begin{algorithm}[ht]
\caption{\ZWalk($L$), agent $\alpha$ moves on the path up to distance $L$ from $v$ on both sides}
\label{alg:doubling-walk}
\begin{algorithmic}[1]
\State Agent $\alpha$ takes $L$ steps in one direction.
\State Agent $\alpha$ takes $2L$ steps in the opposite direction.
\State Agent $\alpha$ takes $L$ steps in the first direction, getting back to its original position.
\end{algorithmic}
\end{algorithm}

\begin{algorithm}[ht]
\caption{Rendezvous algorithm for agent $\alpha$ starting at node $v_\alpha$}
\label{alg:main-algorithm}
\begin{algorithmic}[1]
    \Statey Let $\ESPRSdl > 0$ be a universal constant s.t.\ any node $v$ executing $\EarlyStopPathRulingSet(R)$\ computes its output in at most $\ESPRSdl R \log^* \ell(v)$ rounds.
    \State Initialize $L \gets 1$.
    \While{not(RendezVous)}
\State \ZWalk($L$). \Comment{Discovery phase}

    \State For each $R \in \set{4^j \mid j\in \naturals}$, let $S_R = \set{u  \in N^R(v_\alpha), N^{\ESPRSdl R \log^* (u)}(u) \subseteq N^L (v_\alpha)}$

    \If{$\exists R \in \set{4^j \mid j\in \naturals}$, $R\leq L/16$, s.t.\ $S_R \neq \emptyset$}
        \State Let $\Rstar = \max \set{R \in \set{4^j \mid j\in \naturals}: R \leq L/16 \wedge S_R \neq \emptyset}$.

        \State Simulate \EarlyStopPathRulingSet$(\Rstar)$ for all nodes $u \in S_\Rstar$.
    
        \State Let $r$ be the known colored ruling set node closest to $v_\alpha$, breaking ties with IDs.

        \State $\SearchingWalk(\Rstar,L,r,\chi(r))$. \Comment{Searching phase}
    \Else   \Comment{Not enough nodes known to run \EarlyStopPathRulingSet}
         \State    Wait for $24L$ rounds.
    \EndIf

\State $L \gets 2L$.
\EndWhile
\end{algorithmic}
\end{algorithm}

Our algorithm (\cref{alg:main-algorithm}) consists of a loop with a parameter $L$ that doubles between iterations. Inside the loop are two phases, which we call the \emph{discovery phase} and the \emph{searching phase}. In the discovery phase, an agent explores the distance-$L$ neighborhood around their starting position, roughly doubling how much of the network they know about. The agents use this information to compute a colored ruling set of parameter $R$, where $R \leq L/16$ is the largest power of $4$ s.t.\ the distance-$O(R)$ neighborhood of their starting location contains a ruling-set node. Of course, as the agents have different views, they do not necessarily end up computing a ruling-set with the same parameter $R$, and we denote by $R_{\alpha,L}$ and $R_{\beta,L}$ the distance of the ruling sets computed by agents $\alpha$ and $\beta$ in their phase $L$. The agents use their respective computed ruling set to search for one another, using procedure \SearchingWalk\ (\cref{alg:searching-walk}), which constitutes the searching phase.

In \SearchingWalk\ (\cref{alg:searching-walk}), each agent goes to a nearby colored node of the computed ruling set.
Once at this node, the agents alternate calls to \ZWalk\ and waiting periods, according to the color of the ruling set node.
The fact that close enough ruling set nodes have distinct colors guarantees that the two agents, if they have computed a ruling set with the same distance $R \geq D$, will have different behaviors at some point. That is, one of the agents will be performing a \ZWalk\ while the other is not moving. The agents achieve rendezvous as this happens.

\begin{algorithm}[ht]
\caption{$\SearchingWalk(R,L,r,c)$, in which agent $\alpha$ searches in the distance-$8R$ neighborhood of node $r$ close to its initial position $v_\alpha$}
\label{alg:searching-walk}
\begin{algorithmic}[1]
\State Move to node $r$. \Comment{at distance $\leq 2R-1$ from $v_\alpha$}
\State Wait for $L - \dist(v_\alpha,r)$ rounds. \Comment{$\geq 0$ since $L \geq 16R$}
\State \ZWalk$(8R)$. \Comment{$32R$ rounds}
\For{$i$ from $0$ to $4$}  \Comment{iterating over the bits of $c = \chi(r)\in \set{0,1}^5$}
\If{$c_i=1$}
\State \ZWalk$(8R)$ twice. \Comment{$2 \times 32R$ rounds}
\Else \Comment{$c_i=0$} 
\State Wait for $64R$ rounds.
\EndIf
\EndFor
\State Wait for $11\cdot(2L - 32R)$ rounds. \Comment{$\geq 0$ since $L \geq 16R$}
\State Wait for $L - \dist(v_\alpha,r)$ rounds.
\State Move back to node $v_\alpha$.
\end{algorithmic}
\end{algorithm} 

\subsection{Analysis Overview}

Missing from the previous section's rough overview is notably how the argument is robust to the fact that agents can have a delay of $\tau$ between their wake-up times.
We now explain how, in the next sections, we analyze our algorithm and prove \cref{thm:main-result}, that is, that the agents achieve rendezvous in $O(D \log^* \lmin)$ rounds when executing \cref{alg:main-algorithm}.

\subparagraph{Organization of the proof.}
To prove \cref{thm:main-result}, we first bound in \cref{sec:distinct-phase-rdv} how out-of-sync the two agents can be without trivially achieving rendezvous. Knowing that a delay $\tau$ that exceeds some threshold $\Theta(D)$ necessarily implies that the agents achieve rendezvous in $O(D)$ rounds, the rest of the analysis can assume some degree of synchronization between the agents.

We then argue that as the agents reach a large enough value of $L$ (of order $\Theta(D \log^* \lmin)$) in their executions of \cref{alg:main-algorithm}, they compute nodes from colored ruling sets parameterized by distances of order at least $\Omega(D)$. 
When using such ruling set nodes in their executions of \SearchingWalk\ (\cref{alg:searching-walk}), the agents can achieve rendezvous in two different ways. If the two agents have computed nodes from distinct ruling sets (i.e., the distance parameters of their ruling sets do not match), the agent with the larger distance is guaranteed to find the other agent.
This follows from the agent with the smaller distance parameter being constrained to some interval on the line, which we show gets fully explored by the agent with the larger parameter.
We show this in \cref{sec:mismatched-set-rdv}. The analysis for when the two agents computed nodes from the same ruling set (i.e., with the same distance parameter) is in \cref{sec:common-set-rdv}. There, we argue that rendezvous is easily achieved if the two agents decided to go to the same node of the ruling set, and if not, they find each other from the fact that they perform \SearchingWalk\ from nodes that are both close enough to one another, and are colored with distinct colors.
The various arguments are put together in \cref{sec:algorithm-analysis}.

\subparagraph{Notation and preliminary remarks.}
We use the following notation:
\begin{itemize} 
    \item For each $L\in \set{2^i \mid i \in \naturals}$, $T_{\alpha,L}$ is the round in which agent $\alpha$ starts its discovery phase with parameter $L$ (line $3$ in \cref{alg:main-algorithm}). $T_{\beta,L}$ is the equivalent for agent $\beta$.
    \item $\Trdv \in \naturals \cup \set{\infty}$ is the round in which the agents first achieve rendezvous ($\infty$ if the agents never do so). 
    \item As previously in the paper, we identify the infinite line with the relative integers. For each $t\in \naturals$, let $x(t)$ (resp.\ $y(t)$) be the position of agent $\alpha$ (resp.\ $\beta$) in round $t$. We assume $\alpha$ to be the first agent to wake up, so that $y(t) = y(0) = v_\beta$ for all $t \leq \tau$.
    \item For each $L\in \set{2^i \mid i \in \naturals}$, let $R_{\alpha,L} \in (\set{4^i \mid i \in \naturals}\cap [1,L/16]) \cup \set{\bot}$ be the value of $\Rstar$ for agent $\alpha$ for the loop iteration of parameter $L$. $R_{\beta,L}$ represents the same quantity for $\beta$, and the value $\bot$ represents the fact that no valid $\Rstar$ was found for this loop iteration.  
\end{itemize}

A useful remark about our algorithm is that the discovery and searching phases of our agents have fixed runtimes. 

\begin{proposition}
    \label{prop:phase-runtimes}
    For each $L \in \set{2^i \mid i \in \naturals}$, we have:
    \begin{enumerate}
        \item $T_{\beta,L} = T_{\alpha,L} + \tau$,
        \item $T_{\alpha,L} = 28(L-1)$.
    \end{enumerate}
\end{proposition}
\begin{proof}
    The loop iteration of \cref{alg:main-algorithm} with a given value of $L$ takes $28L$ rounds. Indeed, the discovery phase is a \ZWalk\ of parameter $L$, taking $4L$ rounds. Whether some $\Rstar = 4^j$ is found such that
    $\exists u  \in N^\Rstar(v_\alpha), N^{\ESPRSdl \Rstar \log^* (u)}(u) \subseteq N^L (v_\alpha)$ or not, the rest of the loop always takes $24L$ rounds. Either these rounds are spent waiting at $v_\alpha$, or they are spent executing \SearchingWalk\ (\cref{alg:searching-walk}). \SearchingWalk\ has this complexity that is a function of $L$ only because each of its actions  whose runtime increases with $R$ is compensated by a waiting period whose time decreases with $R$ by the exact same amount. $2\times L$ rounds are spent traveling and waiting between $v_\alpha$ and $r$, and $11\times 2L$ are spent in $11$ calls to $\ZWalk(8R)$ and/or waiting periods, for a total of $24L$ rounds.

    The results follow from the just established formulas $T_{\alpha,2L} = T_{\alpha,L} + 28L$ and $T_{\beta,2L} = T_{\beta,L} + 28L$, together with the base case $T_{\beta,1} = T_{\alpha,1} + \tau$.
\end{proof}

\subsection{First Argument: Out-of-Sync Rendezvous}
\label{sec:distinct-phase-rdv}

In this section, we show that the agents achieve rendezvous relatively fast and easily when the delay $\tau$ is large enough. For a high-level intuition, assume that $\tau$ is very large compared to $D$. A property of our algorithm is that after $t$ rounds of being active, an agent has explored a neighborhood of size $\Theta(t)$ around its starting location.
With a distance of $D$ between the two agents' starting locations, this means that if after $\Theta(D)$ rounds of being active, agent $\alpha$ has reached the starting location of agent $\beta$. If $\beta$ is still asleep as this happens, we have that agent $\alpha$ finds it as it reaches $v_\beta$ for the first time. We essentially show an extension of this for smaller values of $\tau$.

We make some simple observations about the positions of the agents, from which it immediately follows that a delay of $10D$ or more makes the agents achieve rendezvous in $O(D)$ rounds.

\begin{proposition}
    \label{prop:agent-confinement-and-coverage}
    Let $L = 2^i$ for some integer $i$.
    \begin{enumerate}
        \item For all $t \leq T_{\alpha,L} + L/2$, $\abs{x(t)-v_\alpha} \leq L/2$.
        \item For all $z \in [v_\alpha-L,v_\alpha+L]$, $\exists t \in [T_{\alpha,L},T_{\alpha,L}+3L]$ s.t.\ $x(t) = z$.
    \end{enumerate}
\end{proposition}
\begin{proof}
    For the first inequality, first note that by definition $x(T_{\alpha,L}) = x(0) = v_\alpha$ since the agent always returns at $v_\alpha$ before starting a new loop iteration. This immediately implies that $\abs{x(t)-x(0)} \leq L/2$ for $t \in [T_{\alpha,L},T_{\alpha,L}+L/2]$, as agents can only cross one edge per round.
    For smaller $t$ ($t \in [0,T_{\alpha,L}]$, let us consider the movements done in loop iterations for smaller values of $L$. As an agent executes a \ZWalk\ with some parameter $d$, it stays at distance $\leq d$ from the node where it started its execution. As a results, all calls to \ZWalk\ as part of previous discovery phases have kept the agent at distance $\leq L/2$ from $x(0)=v_\alpha$.
    As for the movements done in the searching phase, the agent moves at distance at most $2\Rstar-1+8\Rstar < L$ from $v_\alpha$. Indeed, at most $2\Rstar-1$ are done to go to a nearby ruling set node, from which \ZWalk\ is called with parameter $8\Rstar$, and $\Rstar \leq L/16$.
    
    The second item simply follows from the fact that a \ZWalk\ of parameter $L$ consists of going $L$ steps away from the starting position, then taking $2L$ steps in the other direction. Therefore, either $x(T_{\alpha,L}+L) = v_\alpha - L$ and $x(T_{\alpha,L}+3L) = v_\alpha + L$, or $x(T_{\alpha,L}+L) = v_\alpha + L$ and $x(T_{\alpha,L}+3L) = v_\alpha - L$. In either case, the entire interval $[v_\alpha-L,v_\alpha+L]$ is swept.
\end{proof}

\begin{lemma}
    \label{lem:out-of-sync-rendezvous}
    Let $L = 2^{\ceil{\log D}+1}$. Either $\tau < 5L/2 \leq 10D$, or $\Trdv \leq T_{\alpha,L} + 3L \in O(D)$.
\end{lemma}
\begin{proof}
    Note that $L \in [2D,4D-1]$. If $\tau \geq 5L/2$, then for all $t \leq  T_{\alpha,L} + 3L \leq T_{\alpha,L} + \tau + L/2 = T_{\beta,L} + L/2  $, $\abs{y(t) - v_\beta} \leq L/2$ by \cref{prop:agent-confinement-and-coverage}.
    Because $\abs{v_\alpha - v_\beta} = D$, this also implies that $\abs{y(t) - v_\alpha} \leq D + L/2 \leq L$ for all $t \leq T_{\alpha,L} + 3L$.
    Since $\alpha$ explores the entire distance-$L$ neighborhood of $v_\alpha$ between $T_{\alpha,L} + L$ and $T_{\alpha,L} + 3L$ by \cref{prop:agent-confinement-and-coverage}, the two agents achieve rendezvous before round $T_{\alpha,L} + 3L$ by \cref{prop:no-position-exchange}. 
\end{proof}

\subsection{Second Argument: In-Sync Rendezvous}
\label{sec:same-phase-rdv}

Suppose now that the delay between the agents is less than $10D$. This implies that for sufficiently large values of $L$ (greater than some $\Theta(D)$), the discovery and searching phases of parameter $L$ of the two agents mostly overlap. We argue now that in this case, the agents can rely on the searching phase to achieve rendezvous, once $L$ has reached a large enough value (of order $\Theta(D \log^* \lmin)$).

Recall that $R_{\alpha,L}$ and $R_{\beta,L}$ are the values $\Rstar$ computed by the two agents in the loop iteration of parameter $L$.
$R_{\alpha,L}$ is the maximal value $R \in \set{4^i \mid i \in \naturals} \cap [0,L/16]$ such that for some node $u$ within distance $R$ from $v_\alpha$, the output of $u$ in $\EarlyStopPathRulingSet(R)$ can be inferred from nodes within distance $L$ from $v_\alpha$. $R_{\alpha,L}=\bot$ indicates that no such value $R$ exists for the loop iteration $L$. $R_{\beta,L}$ is the same variable for agent $\beta$.

The output of a node $u$ in $\EarlyStopPathRulingSet(R)$ includes a node from a $(R,R-1)$-ruling set within distance $R-1$ from $u$ (possibly $u$ itself), by \cref{lem:earlystoppathrulingset}. In a loop iteration where $R_{\alpha,L}\neq \bot$, this means that $\alpha$ has learned about a node from the $(R_{\alpha,L},R_{\alpha,L}-1)$-ruling set computed by algorithm $\EarlyStopPathRulingSet(R_{\alpha,L})$ at distance at most $2R_{\alpha,L}-1$ from $v_\alpha$.
We denote $r_{\alpha,L}$ this node, and $c_{\alpha,L}$ the color it got imparted in $\EarlyStopPathRulingSet(R_{\alpha,L})$.

Recall that $\lmin$ is the smallest label within distance $D$ from one of the agents' starting locations, i.e., 
$\min\set{\ell(u) \mid u \in N^D(v_\alpha) \cup N^D(v_\beta)}$,
and that $\ESPRSdl$ is a universal constant s.t.\ a node $u$ of label $\ell(u)$ executing $\EarlyStopPathRulingSet(R)$ computes its output in at most $\ESPRSdl R \log^* \ell(u)$ rounds.

\begin{lemma}
    \label{lem:guaranteed-ruling-set}
    If $L \geq 8\ESPRSdl D \log^* \lmin + 128D$ and $L \in \set{2^i \mid i \in \naturals}$, we have $R_{\alpha,L} \geq 2D$ and $R_{\beta,L} \geq 2D$.    
\end{lemma}
\begin{proof}
    Let $v_{\min} \in N^D(v_\alpha) \cup N^D(v_\beta)$ be the node s.t.\ $\ell(v_{\min}) = \lmin$ and $R_D = 4^{\ceil{\log_4 (2D)}}$. Note that $2D \leq R_D < 8D$, $\dist(v_{\min},v_\alpha) \leq 2D \leq R_D$, and $\dist(v_{\min},v_\beta) \leq 2D \leq R_D$.

    Consider lines 5 and 6 in \cref{alg:main-algorithm}. $R_{\alpha,L}$ is the largest value that satisfies the following three properties: $R_{\alpha,L}$ is a power of $4$, $R_{\alpha,L} \leq L/16$, and agent $\alpha$ has discovered enough of the line in its discovery phase of parameter $L$ to compute the output of at least one node within distance $R_{\alpha,L}$ from $v_\alpha$ in an execution of $\EarlyStopPathRulingSet(R_{\alpha,L})$. $R_{\beta,L}$ is defined similarly for agent $\beta$. We show that all these properties hold for $R_D = 4^{\ceil{\log_4 (2D)}}$ with the lower bound on $L$ given as hypothesis, which implies the lemma.

    $R_D$ is a power of $4$ by definition. Also, since $L \geq 128D$ and $R_D < 8D$, we have $R_D \leq L/16$. Node $v_{\min}$ is at distance at most $2D \leq R_D$ from $v_\alpha$. Finally, consider the distance-$(\ESPRSdl R_D \log^* \lmin)$ neighborhood of $v_{\min}$. From the bound on the distance between $v_{\min}$ and $v_\alpha$, and $R_D < 8D$, it holds that $N^{\ESPRSdl R_D \log^* \lmin}(v_{\min}) \subseteq N^{8\ESPRSdl D \log^* \lmin + 2 D}(v_{\alpha}) \subseteq N^L(v_\alpha)$. This implies that $S_{R_D} \neq \emptyset$ for $S_{R_D}$ defined as in lines 5 and 6 of \cref{alg:main-algorithm}, and therefore, $R_{\alpha,L} \geq R_D \geq 2D$. The same holds for agent $\beta$.
\end{proof}

\subsubsection{Mismatched Ruling Sets Rendezvous}
\label{sec:mismatched-set-rdv}

We first consider the case where one of the agents has computed part of a ruling set of higher distance $\Rstar$ than the other agent. As $R_{\alpha,L}$ and $R_{\beta,L}$ are powers of $4$, $R_{\alpha,L} > R_{\beta,L}$ implies that $R_{\alpha,L} \geq 4 R_{\beta,L}$. This large difference between the two distances allows us to prove rendezvous by similar arguments as in the previous section (\cref{lem:out-of-sync-rendezvous}).

\begin{lemma}
    \label{lem:mismatched-rendezvous}
    Suppose that $R_{\alpha,L} \neq R_{\beta,L}$ and $\max(R_{\alpha,L},R_{\beta,L}) \geq D$. Then $\Trdv \leq T_{\alpha,L} + 7L \in O(L)$.
\end{lemma}
\begin{proof}
    Assume $\tau \leq 10D$ as otherwise \cref{lem:out-of-sync-rendezvous} already yields the claim.
    
    Let agent $\alpha$ have computed the ruling set of higher distance, i.e., $R_{\alpha,L} > R_{\beta,L}$ and $R_{\alpha,L} \geq D$.
    Note that $R_{\alpha,L} \leq L/16$, so $L \geq 16D \geq 8\tau/5$.
    In the first $L$ rounds of \SearchingWalk\ (rounds $T_{\alpha,L} + 4L$ to $T_{\alpha,L} + 5L$), $\alpha$ moves to the ruling set node $r_{\alpha,L}$ at distance at most $2R_{\alpha,L}-1$ from $v_\alpha$. It then performs a \ZWalk\ of parameter $8R_{\alpha,L}$.
    As a result, all the nodes within distance $8R_{\alpha,L} - (2R_{\alpha,L}-1) - D > 5R_{\alpha,L}$ from $v_\beta$ are explored by $\alpha$ between round $T_{\alpha,L} + 5L$ and $T_{\alpha,L} + 6L$.
    
    Consider agent $\beta$ during the same rounds. We know that $\beta$ started its discovery phase of parameter $L$ in round $T_{\beta,L} = T_{\alpha,L} + \tau \leq T_{\alpha,L} + L$. Therefore, between round $T_{\alpha,L} + 5L$ and $T_{\alpha,L} + 6L$, $\beta$ has finished its discovery phase of parameter $L$ and is either staying idle at $v_\beta$ or executing \SearchingWalk. Either way, $\beta$ remains within distance $(2R_{\beta,L}-1) + 8R_{\beta,L} \leq (5/2)R_{\alpha,L}$ from $v_\beta$ in those rounds.
    
    Therefore, between round $T_{\alpha,L}+5L$ and $T_{\alpha,L}+6L$, $\beta$ stays within an interval of nodes while $\alpha$ fully explores that interval of nodes. By \cref{prop:no-position-exchange}, the agents achieve rendezvous by round $T_{\alpha,L} + 6L$.
    The case $R_{\beta,L} > R_{\alpha,L}$ is similar, only introducing an extra $\tau \leq 10D \leq L$ to the number of rounds needed to achieve rendezvous.  
\end{proof}

\subsubsection{Common Ruling Set Rendezvous}
\label{sec:common-set-rdv}

The only case that remains to analyze is that where $R_{\alpha,L} = R_{\beta,L}$. We first show that if the agents go to the same node from the ruling set, they easily achieve rendezvous.
Otherwise, they achieve rendezvous as they perform successive calls to \ZWalk\ according to the (distinct) colors given to $r_{\alpha,L}$ and $r_{\beta,L}$ by \EarlyStopPathRulingSet.

\begin{figure}[h]
        \centering
        \includegraphics[page=2,width=0.8\textwidth]{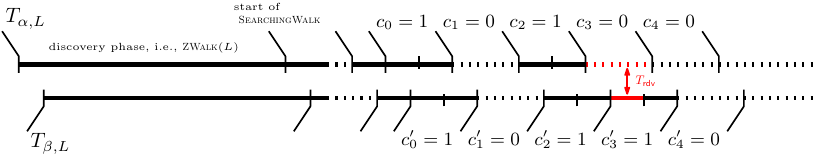}
        \caption{Timeline of an execution of \SearchingWalk\ by the two agents, with colors $c=\mathtt{10100}$ and $c' = \mathtt{10110}$. Horizontal solid lines represent movement (mainly, calls to \ZWalk) while dotted lines represent waiting. Rendezvous occurs during the processing of bit $i=3$.}
        \label{fig:searching-walk}
\end{figure}

\begin{lemma}
\label{lem:same-ruling-node-rendezvous}
    Suppose that $R_{\alpha,L} = R_{\beta,L} \geq D$ and $r_{\alpha,L} = r_{\beta,L}$. Then $\Trdv \leq T_{\alpha,L} + 5L \in O(L)$.
\end{lemma}
\begin{proof}
    As before we assume that $\tau \leq 10D$, since \cref{lem:out-of-sync-rendezvous} otherwise already yields the claim. As $D \leq R_{\alpha,L}$, and $R_{\alpha,L} \leq L/16$, this implies $\tau \leq 5L/8$.
    
    The agents both take at most $2R_{\alpha,L}-1 \leq L/8$ rounds from the start of their respective \SearchingWalk\ to reach the node $r_{\alpha,L} = r_{\beta,L}$. Therefore, between rounds $T_{\alpha,L}+4L+L/8$ and $T_{\alpha,L}+5L$, $\alpha$ is at $r_{\alpha,L} = r_{\beta,L}$.
    $\beta$ reaches this node by round $T_{\beta,L}+4L+L/8 \leq T_{\alpha,L}+\tau + 4L+L/8 \leq T_{\alpha,L} + 4L +3L/4$ at the latest. Thus the agents necessarily meet at $r_{\alpha,L} = r_{\beta,L}$ during the beginning of \SearchingWalk\ before round $T_{\alpha,L} + 5L$.
\end{proof}

Finally, we tackle the last case of our analysis: the agents have computed nodes from the same ruling set, and they achieve rendezvous using the colors that were assigned to these nodes by \EarlyStopPathRulingSet.

\begin{lemma}
\label{lem:distinct-nodes-same-set-rendezvous}
    Suppose that $R_{\alpha,L} = R_{\beta,L} \geq D$ and $r_{\alpha,L} \neq r_{\beta,L}$. Then $\Trdv \leq T_{\alpha,L} + 28L \in O(L)$.
\end{lemma}
\begin{proof}
    Again let $\tau \leq 10D$ (\cref{lem:out-of-sync-rendezvous}), and note that $\tau \leq 10R_{\alpha,L} \leq L$.

    Consider $c_{\alpha,L}$ and $c_{\beta,L}$ the colors assigned to the nodes $r_{\alpha,L}$ and $r_{\beta,L}$ by the execution of $\EarlyStopPathRulingSet(R_{\alpha,L})$.
    As we have that $\dist(v_\alpha,r_{\alpha,L}) \leq 2R_{\alpha,L}-1$, $\dist(v_\beta,r_{\beta,L}) \leq 2R_{\alpha,L}-1$, and $\dist(v_\alpha,v_\beta) \leq D \leq R_{\alpha,L}$, it holds that nodes $r_{\alpha,L}$ and $r_{\beta,L}$ are at distance $\dist(r_{\alpha,L},r_{\beta,L}) \leq 5R_{\alpha,L} - 2$.
    Algorithm $\EarlyStopPathRulingSet(R)$ computes a colored $(R,R-1)$-ruling set such that nodes of the ruling set get distinct colors when at distance $9R-1$ or less from one another (\cref{lem:earlystoppathrulingset}). Therefore, $c_{\alpha,L} \neq c_{\beta,L}$.

    Let $i$ be the first bit where $c_{\alpha,L}$ and $c_{\beta,L}$ differ. Let us assume that $\alpha$ has a $1$ at index $i$ of its color while $\beta$ has a $0$, the opposite case being similar.
    Between round $T_{\alpha,L}+5L + (2i+1)\cdot 32 R_{\alpha,L}$ and $T_{\alpha,L}+5L + (2i+3)\cdot 32R_{\alpha,L}$, agent $\alpha$ is executing $\ZWalk(8 R_{\alpha,L})$ twice in succession. 
    In particular, $\alpha$ scans the entire distance-$8R_{\alpha,L}$ neighborhood of $r_{\alpha,L}$ between round $T_{\alpha,L}+5L + (2i+2)\cdot 32 R_{\alpha,L}$ and $T_{\alpha,L}+5L + (2i+3)\cdot 32 R_{\alpha,L}$.
    Recall that $\tau \leq 10D \leq 10 R_{\alpha,L}$ and $T_{\beta,L} = T_{\alpha,L} + \tau$. Agent $\beta$ is staying at $r_{\beta,L}$ for $64 R_{\alpha,L}$ rounds starting from round $T_{\beta,L}+5L + (2i+1)\cdot 32 R_{\alpha,L} \leq T_{\alpha,L}+5L + (2i+1)\cdot 32 R_{\alpha,L} + 10 R_{\alpha,L}$. Therefore, $\alpha$ necessarily finds $\beta$ during its second \ZWalk\ for bit $i$, by \cref{prop:no-position-exchange}.
    See \cref{fig:searching-walk} for an illustration with $i=3$ as first index where their colors differ.
    The bound on $\Trdv$ simply follows from the runtime of \SearchingWalk.
\end{proof}

\subsection{Full Analysis}
\label{sec:algorithm-analysis}

We now have analyzed all the cases necessary to show \cref{thm:main-result}.

\MainResultThm*

\begin{proof}
    Recall that $T_{\alpha,L} \in O(L)$ (\cref{prop:phase-runtimes}).
    If the delay $\tau$ between the startup times of the two agents is $10D$ or larger, by \cref{lem:out-of-sync-rendezvous}, the agents reach rendezvous in $T_{\alpha,O(D)} + O(D) \in O(D)$ rounds ($\Trdv \in O(D)$).
    We assume a smaller delay $\tau$ in what follows.

    By \cref{lem:guaranteed-ruling-set}, we know that $R_{\alpha,L}\geq D$ and $R_{\beta,L}\geq D$ for $L > \Theta(D \log^* \lmin)$.
    For $L$ a large enough $\Theta(D \log^* \lmin)$, if during the loop iteration of parameter $L$ the agents settle on different values $R_{\alpha,L} \neq R_{\beta,L}$ for their ruling sets' distances, they achieve rendezvous before round $T_{\alpha,L} + O(L) = O(L) = O(D \log^* \lmin)$ by \cref{lem:mismatched-rendezvous}.
    If they have the same value $R_{\alpha,L} = R_{\beta,L}$, then they also achieve rendezvous before round $T_{\alpha,L} + O(L) = O(L) = O(D \log^* \lmin)$, whether they go to the same ruling set node (\cref{lem:same-ruling-node-rendezvous}) or different ones (\cref{lem:distinct-nodes-same-set-rendezvous}).
\end{proof}

\section{Corollary for Finite Graphs}
\label{sec:other-results}

In this section, we give simple extensions of our result for the infinite line to analogous path-like graphs like finite paths and cycles.
Adapting our algorithm to also handle these graphs is only the matter of introducing some simple new behavior when an agent detects that the graph is finite (e.g., when it finds a node of degree $1$ for a finite path). 
We prove \cref{thm:finite-graph-main}.

\FiniteGraphThm*

\begin{proof}
We do the following adjustment to our algorithm for the infinite line:
\begin{enumerate}
    \item While an agent has neither seen a node of degree $1$ (the end of a finite path), or reached the same node from two distinct edge-disjoint paths from its starting position (indicating a cycle), it simply executes the algorithm for the infinite line.
    \item If an agent finds itself at a node of degree $1$, it walks straight in the direction away from that node (using \CarefulWalk)
    \item If an agent detected that it is in a cycle, it goes to the node of minimum label in the graph, and waits there.
\end{enumerate}

Suppose $n \in O(D \log^* \lmin)$. After $\Theta(\min(n, D \log^* \lmin))$ rounds, if rendezvous has not yet occurred, the two agents have discovered the entire graph.
In particular, they are both awake.
For the finite line, this already implies that the agents achieved rendezvous by \cref{prop:no-position-exchange}.
For the cycle, the agents are guaranteed to meet once they have both reached the node of minimum ID, which must occur after $\Theta(\min(n, D \log^* \lmin))$ rounds (both agents must be awake by the time one of the two agents detects the cycle).

If $n \geq c\cdot D \log^* \lmin$ for a sufficiently large constant $c$, then the part of the graph that is explored by the two agents in $(c/8)\cdot D \log^* \lmin$ is indistinguishable from a segment from an infinite line.
In that case, the agents meet in $O(D \log^* \lmin)$ rounds from our analysis for the infinite line.
\end{proof}

On finite lines and the infinite half-line, the time to rendezvous also decreases with the distance of the agents to the nearest endpoint (node of degree $1$) of the line.

\begin{remark}
\label{rem:endpoint-distance}
    On the finite labeled line and the infinite labeled half-line, rendezvous can be achieved in $O(\min(d+D, D \log^* \lmin))$ rounds, where $d = \min_{u \in \deg^{-1}(1)} \min_{v\in\set{v_\alpha,v_\beta}} \dist(u,v)$ is the minimum distance between the agents' starting locations and the graph's endpoint(s).
\end{remark}

\section*{Acknowledgments}
We thank Sebastian Brandt for interesting discussions and anonymous reviewers for their valuable feedback. This research was conducted while Alexandre Nolin was at the CISPA Helmholtz Center for Information Security.

\newcommand{\etalchar}[1]{$^{#1}$}

\end{document}